\documentclass[letterpaper,10pt,journal,twoside]{IEEEtran}
\usepackage{cite}
\usepackage{amsmath,amssymb,amsfonts,mathtools}
\usepackage{algorithmic}
\usepackage{graphicx}
\usepackage{algorithm,algorithmic}
\usepackage{hyperref}
\hypersetup{hidelinks=true}
\usepackage{textcomp}
\usepackage[tight]{subfigure}
\usepackage{arydshln}
\usepackage[varvw,vvarbb]{newtxmath}
 \usepackage{mathrsfs}
\DeclareMathOperator\diag{diag}
\DeclareMathOperator\im{Im}

\DeclareMathOperator\spec{spec}
\DeclareMathOperator\rnk{rank}

\DeclareMathOperator\nrnk{n\hspace{.5pt}rank}
 \newcommand\wbar\bar
\newcommand\abs[1]{\ensuremath{\lvert#1\rvert}}

\newcommand\norm[1]{\ensuremath{\lVert#1\rVert}}
\newcommand\onesn{\ensuremath{\mathbb1_{\nu}}}

\newcommand\E{\mathrm e} % e
\newcommand\J{\mathrm j} 
\newcommand\ones{\ensuremath{\mathbb1}} % all ones
\newcommand\Hinf{\ensuremath{H_{\infty}}}

\newcommand\RHP{\ensuremath{\mathbb{C}_0}}
\newcommand\CRHP{\ensuremath{\wbar{\mathbb{C}}_0}}

\newcommand{\Gr}{\mathcal{G}} % graph
\newcommand\Adj{A_{\Gr}} % Adjacency matrix
\newcommand\Deg{D_{\Gr}} % Degree matrix
\newcommand\Lap{L_{\Gr}} % Laplacian matrix
\newcommand\Mtil{\ensuremath{\tilde{M}}}
\newcommand\Ntil{\ensuremath{\tilde{N}}} % LCF \Ntil^{-1}\Ntil
\newcommand\Xtil{\ensuremath{\tilde{X}}}
\newcommand\Ytil{\ensuremath{\tilde{Y}}} % \Mtil \Xtil+\Ntil \Ytil =I
\newcommand\mmatrix[2][cccccc]{\ensuremath{\left[\begin{array}{#1}#2\end{array}\right]}}
\makeatletter
\xdef\@endgadget#1{{\unskip\nobreak\hfil\penalty50\hskip1em\hbox{}\nobreak\hfil#1\parfillskip=0pt\finalhyphendemerits=0\par}}
\newcommand\@Endofsymbol{$\triangledown$}
\newcommand\Endofremark{\@endgadget{\@Endofsymbol}}
\makeatother
\newtheorem{theorem}{Theorem}[section]
\newtheorem{lemma}[theorem]{Lemma}
\newtheorem{corollary}[theorem]{Corollary}

\newtheorem{remark}{Remark}[section]
\newtheorem{example}{Example}[section]
\newtheorem{definition}{Definition}[section]
\begin{document}
%
%\title{On the Fragility and Robustness of Synchronization Trajectories to Network Perturbations}
\title{How network perturbations distort agreement trajectories in LTI multi-agent systems}
%\title{On the robustness of synchronized trajectories to network perturbations: a frequency domain approach}
%
\author{Gal Barkai$^1$, \and Irinel-Constantin Mor\u{a}rescu$^1$ 
\thanks{$^1$Universit\'e de Lorraine, CNRS, CRAN, F-54000 Nancy, France, Emails: \textsl{\{gal.barkai,constantin.morarescu\}@univ-lorraine.fr}.}
}
%\thanks{*This work has been supported in part by ANR under grant COMMITS ANR-23-CE25-0005.} 
\maketitle

\begin{abstract}
Distributed coordination of multi-agent systems frequently relies on cooperative protocols designed to achieve agreement on a prescribed, non-trivial trajectory. While the robustness of such protocols to various uncertainties is well documented, existing literature universally assumes that the target agreement trajectory itself remains invariant. This assumption may hold in ideal cases, but we prove that network perturbations can vastly modify the asymptotic agreement trajectory. We first investigate the exact trajectories of Linear Time-Invariant (LTI) agents subjected to dynamic coupling uncertainties by establishing a new Laplace-domain criterion that characterizes the specific closed-loop poles governing the perturbed agreement manifold. To formalize our analysis, we introduce the notion of \emph{structure-preserving} dynamics, perturbations that maintain the null space of the communication graph's Laplacian, and contrast them with transmission only dynamics, affecting only the adjacency matrix. We prove a critical fragility within standard cooperative output regulation schemes: while static consensus is uniquely robust to heterogeneous transmission delays, synchronization to periodic trajectories is destroyed by arbitrarily small transmission delays. Furthermore, we demonstrate that for $d$-regular topologies, uniform transmission perturbations can easily shift the system to synchronize with an unexpected, entirely new frequency. These findings expose a previously unidentified vulnerability in classical robust synchronization, demonstrating that transmission dynamics necessitate fundamental structural modifications to networked reference generators.
\end{abstract}
\begin{IEEEkeywords}
% Enter key words or phrases in alphabetical order, separated by commas. Using the IEEE Thesaurus can help you find the best standardized keywords to fit your article. Use the thesaurus access request form for free access to the IEEE Thesaurus: \underline{https://www.ieee.org/publications/services/thesaurus-acce}\\
% \underline{ss-page.com.}
Consensus, Network analysis and control,Linear systems, Uncertain systems.
 \end{IEEEkeywords}

\section{Introduction}
\label{sec:intro}
\IEEEPARstart{O}{ver} the past two decades, the distributed control of multi-agent systems (MAS) has emerged as a central theme in the control literature \cite{O-SM:04}. The harbinger of this trend, and still the most studied distributed control problem, is the consensus problem, which seeks to design a distributed control law that asymptotically drives the states of multiple integrator agents to a common constant value. However, consensus is merely a special case of the general agreement or synchronization problem, where the objective is to drive the states or outputs of the agents to a common, time-varying trajectory \cite{SS:09}. For homogeneous agents, agreement is typically achieved by imposing structured coupling based on a matrix representation of the communication topology, effectively reducing the system to a low-dimensional robust stabilization problem \cite{FM:04,LDCH:10,BenRejeb2018GuaranteedCost}. A seminal generalization of these concepts to heterogeneous agents relies on the internal model principle \cite{FrW:76}, establishing that it is necessary and sufficient for all agents to realize a common internal model to synchronize \cite{WSA:11}. This cooperative output regulation approach prescribes the eventual agreement trajectory \emph{a priori} during the control design stage.

Given the networked nature of MAS, significant research has focused on ensuring synchronization under various network dynamics, often modeling communication constraints and imperfections. These can be largely divided into three categories: i) delays (actuation \cite{O-SM:04,MMN:09}, transmission \cite{M:05,SDJ:08}, and internal \cite{QGY:19}), ii)  uncertain edge weights \cite{ZB:17}, and iii) general finite-dimension uncertainties (additive \cite{TTM:13} or on the edges \cite{LC:17}). Transmission perturbations are of critical interest because they are unavoidable when agents exchange data over a network, a fundamental requirement in cooperative output regulation protocols where agents share the internal states of a common reference generator \cite{WSA:11,IMC:14}. Early consensus results indicated a surprising robustness to such delays; for example, consensus in systems of first-order agents is reached even in the presence of arbitrary transmission delays \cite{M:05}, which is not the case for actuation delays \cite{O-SFM:07}. Similarly, consensus of general Linear Time-Invariant (LTI) agents is much more robust to heterogeneous transmission delays then to heterogeneous actuation delays \cite{MPA:10}. Notably, the more general treatments of robust synchronization typically do not consider transmission-only perturbations at all, and preclude any perturbation including a delay element due to the finite-dimensionality assumption.

Despite these advances, existing literature regarding the analysis of protocols subject to network perturbations overwhelmingly focuses on whether \emph{some} form of agreement is reached, largely ignoring the impact of the perturbation on the asymptotic trajectory itself. In fact, the vast majority of existing works that do consider trajectories restrict their attention only a limited set of target trajectories, such as static or velocity consensus. This can be attributed, in part, to technical difficulties in determining the exact trajectory of a perturbed system. For example, delayed networks do not admit finite-dimensional state-space representations, thus analyzing their steady-state behavior in the time-domain is extremely difficult even for first order agents \cite{SDJ:08,MMN:09}. More general setups often resort to analysis in the Laplace domain, where the main tool to analyze asymptotic behavior is the final value theorem, whose use is limited to poles at the origin \cite{MPA:10,Lee:16}. 

In this paper, we consider a standard LTI MAS designed to reach agreement on some prescribed nominal trajectory. We rigorously define and analyze the robustness of this trajectory to network perturbations, determining whether the nominal trajectory remains attainable despite possibly infinite-dimensional, LTI network uncertainties. To this end, we construct a frequency-domain framework to analyze all possible agreement trajectories based on transfer function analysis. We leverage classical notions such as the existence of left and right coprime factors to obtain a novel characterization of all unstable poles capable of contributing to the asymptotic trajectory. We then show that for simple diffusively coupled networks, like those in all previously mentioned works, this condition simplifies even further to an intuitive structural condition on the graph Laplacian. 

Using this structure we can define a distinct class that we call \emph{structure-preserving perturbations}, for which the original synchronization trajectory does not change regardless of the specific perturbation. Remarkably, the common uncertainty structures considered in \cite{TTM:13,LC:17} are shown to be structure-preserving. Conversely, we show that transmission perturbations are \emph{not} structure-preserving, and that the remarkable robustness of consensus to transmission delays is a by-product of the consensus objective and properties of the delays. We prove that for transmission delays, the static consensus trajectory is the \emph{only} robust trajectory independently of the delay value. For any periodic trajectory, we establish a \emph{minimal} transmission delay, proving that arbitrarily small delays destroy synchronization to the prescribed nominal trajectory. Consequently, reference generators relying on networked internal models \cite{WSA:11} must fundamentally alter their protocols to account for inherent communication latency. Moreover, we prove that for a uniform transmission perturbation the nominal trajectory is achievable only if the perturbation satisfies strict gain and phase conditions at every frequency present in the nominal trajectory. Finally, we demonstrate that for $d$-regular topologies, transmission perturbations do not merely destroy synchronization but can shift the agreement trajectory to unexpected, entirely different trajectories.

The remainder of this paper is organized as follows. Section \ref{sec:agrDef} presents the over all problem setup, rigorously defines agreement and agreement trajectory, and presents key technical results. Section \ref{sec:diffusive} further simplifies the previous results to the common case of homogeneous diffusively coupled systems. Section \ref{sec:frag} introduces the notions of fragility and robustness of trajectories, and proceed to analyze the possible agreement trajectories for two classes of Laplacian dynamics: structure-preserving and transmission only dynamics. A discussion on the results from the perspective of both classical and cooperative output regulation theory is presented in \S~\ref{sec:main:disc}. Extensive numerical examples illustrating the various results are given in Section \ref{sec:ex}, while some concluding remarks are given in Section \ref{sec:conc}. Several definitions and technical results about coprime factorizations over $\Hinf$ are provided in the appendix.
%%%%%%%%%%
%\subsubsection*{Notations and preliminaries}
\subsection{Notations and preliminaries}
%%%%%%%%%%
The sets of integers, real and complex numbers, and the open and closed complex right half-plane are denoted by $\mathbb{Z}$, $\mathbb{R}$, $\mathbb{C}$, $\RHP$ and $\CRHP$ respectively. Given a set $\mathcal S$, its cardinality is denoted as $\abs{\mathcal S}$. By $e_i$ we understand the $i$th standard basis vector in $\mathbb R^\nu$ and by $\onesn$, or simply $\ones$ when the dimension is clear from the context, the all-ones vector from $\mathbb R^\nu$. The complex-conjugate transpose of a matrix $M$ is denoted by $M'$. The notation $\diag\{M_i\}$ stands for a block-diagonal matrix with diagonal elements $M_i$. The image (range) and kernel (null) spaces of a matrix $M$ are notated $\im M$ and $\ker M$, respectively. Given two matrices (vectors) $M$ and $N$, $M\otimes N$ denotes their Kronecker product, while $\spec M$ refers to the set of eigenvalues of $M$. 

We denote a simple directed graph with no self loops by $\Gr$, and its adjacency, incidence, (in) degree and (in) Laplacian matrices by $\Adj$, $E$, $\Deg$ and $\Lap$ respectively. The graph Laplacian is defined as $\Lap=\Deg-\Adj$, and can be equivalently written as $\Lap=EE'$ if the graph is undirected and as $\Lap=BE'$ for some binary matrix $B$ if it is directed. For any digraph the Laplacian satisfies $\Lap \ones =0$ implying that $E'\ones=0$ and that $\Adj\ones=\Deg\ones$, moreover if the graph is connected then $\ker \Lap=\im \ones$ and $0$ is a simple eigenvalue.

Given a complex function $f(s)$, we say that a complex number $p$ is a pole of order $n$ if and only if $1/f(p)=0$ and $(s-p)^n f(s)$ is holomorphic and non-zero in a neighborhood of $p$. A pole of order $1$ is called a simple pole. The \emph{normal rank} of a transfer matrix $G(s)$ is defined as $\nrnk(G(s))\coloneqq \max _{s\in \mathbb{C}} \rnk G(s)$.  We denote by $\Hinf^{p\times m}$ the space of $p\times m$ transfer functions which are analytic and bounded in the supremum norm in the open right half-plane. A function $G\in \Hinf$ can be extended for almost all $\omega$ to some function $\tilde{G}$ such that 
    \[
    \tilde{G}(s)= \begin{cases}
        G(s), &s\in \mathbb{C}_0\\
        \lim_{\sigma \to 0^+}G(\sigma+\J \omega), & s=\J\omega
    \end{cases}
    \]
without loss of generality. To simplify the notation we shall use $G(s)$ even when we refer to $\tilde{G}(s)$.

%%%%%%%%%%%%%%%%%%%%%%%%%%%%
\section{Agreement in time and frequency domains} \label{sec:agrDef}
%%%%%%%%%%%%%%%%%%%%%%%%%%%%
Consider a group of $\nu$ linear time-invariant (LTI) agents described by
\begin{equation} \label{eq:hetdyn_i}
    \begin{cases}
        \dot{x}_i=A_ix_i+B_iu_i & x_i(0)=x_{i,0}\\
        y_i = C_ix_i
    \end{cases} \quad i\in [1\ldots,\nu]
\end{equation}
where $x_i\in \mathbb{R}^{n_i}$, $u_i\in \mathbb{R}^{m_i}$, and $y_i\in \mathbb{R}^p$ are the $i$th state, control, and output signals respectively. Note that we allow for $n_i\neq n_j$ and $m_i\neq m_j$, but assume that all outputs have the same dimension. A network is homogeneous if the matrices in \eqref{eq:hetdyn_i} are identical for all $i\in [1,\ldots,\nu]$. The entire group can be described by
\begin{equation} \label{eq:hetdyn_agg}
    % \begin{cases}
    %     \dot{x}=\diag\{A_i\}x+\diag\{B_i\}u & x(0)=x_{0}\\
    %     y = \diag\{C_i\}x
    % \end{cases} 
    \begin{cases}
        \dot{x}=\hat{A}x+\hat{B}u & x(0)=x_{0}\\
        y = \hat{C}x
    \end{cases} 
\end{equation}
where $x$, $u$, and $y$ are the aggregation of their local counterparts and $\hat{M}\coloneqq \diag\{M_i\}$. To preserve the generality of the discussion we do not enforce a particular structure on $u$, but only assume it can be described as a convolution with some causal and LTI kernel $K(t,\tau)$ in the following fashion
\begin{equation}\label{eq:kernel_u}
    u(t) = \int_0^t K(t-\tau)y(\tau)d\tau.
\end{equation}
By defining the control action via \eqref{eq:kernel_u} we allow for a large class of controllers, namely any causal and LTI system which depend only on the aggregate output $y$. In particular, the controllers do not have to admit a state-space representation. If the kernel matrix in \eqref{eq:kernel_u} is defined by the graph Laplacian $\Lap$ of a static digraph $\Gr$, then the system is called \emph{diffusively coupled}. Note that the matrix structure of $K$ can encode the underlying communication topology in various other ways as well. For this general framework, we are now ready to define precisely what is "agreement".
\begin{definition}[Agreement in the time domain] \label{def:timeAg}
We say that agents \eqref{eq:hetdyn_agg} controlled by \eqref{eq:kernel_u} achieve \emph{output agreement} if for all initial conditions
\begin{subequations}\label{eq:timeagreement}
\begin{equation}\label{eq:timeagreement:out}
    \lim_{t\to \infty} \norm{y_i-y_j}=0 \quad \forall i,j,
\end{equation}
and there exists a square matrix $A_0$ with $\spec A_0 \in \J\mathbb{R}$, a matrix $R$ and a vector $q_0$ such that
    \begin{equation}\label{eq:timeagreement:sync}
        \lim_{t\to \infty} \norm{y_i-R\E^{A_0t}q_0}=0 \quad \forall i.
    \end{equation}
\end{subequations}
We then say that the \emph{agreement trajectory} is $q(t)\coloneqq R\E^{A_0t}q_0$, and that an eigenvalue $\lambda$ can \emph{contribute to the trajectory} if $\lambda \in \spec A_0$.
\end{definition}
Definition \ref{def:timeAg} mirrors standard definitions of state and output synchronization, but is slightly more general in two ways. First, it does not specify the particular communication protocol and control structure. Second, it explicitly defines agreement to a \emph{trajectory}. This is particularly important since in practice the specific trajectory is a design constraint, meaning that $A_0$ is often prescribed a-priori. Defining agreement using only \eqref{eq:timeagreement:out} can inadvertently allow for undesirable behavior, such as converging to the origin or to an unbounded trajectory. Moreover, in some robotic application we would like to synchronize on a specific "safe" trajectory. For standard Laplacian-based controllers the agreement trajectory is determined by the common internal model of all the agents \cite{WSA:11}. This model is embedded into the controllers as a homogeneous agreeing system \cite[eq. 10]{WSA:11} communicating via a Laplacian flow. Thus even when analyzing heterogeneous agents, whether or not the agents reach agreement to a specific trajectory is determined by a simple homogeneous system. 

A natural follow-up question is what happens if there is some unmodeled uncertainty within the system. Indeed, much work has been dedicated to robust synchronization, where usually the uncertainty is modeled as additive, finite-dimensional and stable \cite{TTM:13,LC:17} or as various time-delays \cite{O-SM:04,M:05,SDJ:08,MMN:09,QGY:19}. The aforementioned  works, however, consider agreement only in the sense of \eqref{eq:timeagreement:out}, implicitly allowing for $A_0$ and $R$ from \eqref{eq:timeagreement:sync} to change. In simple words, they do not analyze whether the nominal agreement trajectory can be preserved under the uncertainty. Determining the exact trajectory of the perturbed system is particularly difficult for delayed systems since these do not admit a finite-dimensional state-space representation. In contrast, constant time delays admit a simple and well-behaved transfer function. This is the main motivation for the rest of the section, where we try to adapt Definition \ref{def:timeAg} to the frequency domain. As can be seen in the sequel, this will facilitate the analysis of possible agreement trajectories.

%%%%%%%%%%%%%%%%%%%%%%%%%%%%
\subsection{Frequency domain characterization of agreement}
%%%%%%%%%%%%%%%%%%%%%%%%%%%%
%
\begin{definition}[Agreement in the frequency domain] \label{def:freqAg}
Consider a system whose initial condition response can be described in the frequency domain as
\begin{equation}\label{eq:init_laplace}
    y(s)=G(s)x_0
\end{equation}
for some not necessarily rational transfer function $G(s)$. We say that system \eqref{eq:init_laplace} reaches output agreement if for all initial conditions there exists a square matrix $A_0$ with $\spec A_0 \in \J\mathbb{R}$, a matrix $R$ and a vector $q_0$ such that
\begin{equation}\label{eq:freqagreement}
       G(s)x_0=\ones\otimes \left(R(sI-A_0)^{-1}q_0\right)+Q(s)
\end{equation}
with $Q\in \Hinf$. Given a complex point $\lambda$, we say that $\lambda$ \emph{can contribute to the trajectory} if it is a pole of $(sI-A_0)^{-1}$.
\end{definition}
%%%%
Note that both Definitions \ref{def:timeAg} and \ref{def:freqAg} coincide if we identify $G$ from \eqref{eq:init_laplace} with

\[
% G(s)=\diag{C_i}\left(sI-\diag\{A_i\}-\diag\{B_i\}K(s)\diag{C_i}\right)^{-1},
    G(s)=\hat{C}\left(sI_n-\hat{A}-\hat{B}K(s)\hat{C}\right)^{-1}
\]
where $K(s)$ is Laplace transform of the kernel from \eqref{eq:kernel_u} and $n=\sum_{i=1}^{\nu}n_i$. However, Definition \ref{def:freqAg} is the more general of the two, since it applies for any system which admits a transfer function representation, and not just systems with a finite-dimensional state-space representation.

Our next step is to try and characterize the possible poles of $(sI-A_0)^{-1}$ using only information on $G$. To this end we first state and prove the following instrumental Lemma.
\begin{lemma}\label{lem:cancel}
    Consider a square and stable system $G\in \Hinf$ and some nonzero vector $\eta$ of appropriate dimension. For any integer $m\geq 1$ and complex point $s_0\in \bar{\mathbb{C}}_0$ we have
    \[
    G(s) \eta \frac{1}{(s-s_0)^m}\in \Hinf \implies \eta \in \ker G(s_0).
    \]
\end{lemma}
\begin{IEEEproof}
    By assumption $F(s)\coloneqq G(s)\eta/(s-s_0)^m $ is in $\Hinf$, therefore cannot have poles in the closed right half-plane, implying that the singularity at $s_0$ is removable. Let $\eta_i$ and $g_{ij}(s)$ denote the $i$th and $(i,j)$th entries of $\eta$ and $G$ respectively, thus $F(s)$ is given by
    \[
    F(s) = \mmatrix{\frac{\sum_{j=1}^n \eta_jg_{1j}(s)}{(s-s_0)^m}\\ \vdots \\ \frac{\sum_{j=1}^n \eta_jg_{nj}(s)}{(s-s_0)^m}}\coloneqq \mmatrix{f_1(s)\\ \vdots \\ f_n(s)}.
    \]
    By definition the singularity is removable iff
    \[
    \lim_{s\to s_0} (s-s_0)f_i(s)=0 \quad \forall i\in[1,\dots,n],
    \]
    and a necessary condition for each of these scalar limits to hold is that 
    \[
    \lim_{s\to s_0} (s-s_0)^{m}f_i(s)=0 \quad \forall i\in[1,\dots,n]
    \]
    since
    \[
    \lim_{s\to s_0}\frac{1}{(s-s_0)^{m}} \to \infty.
    \]
    By assumption $G\in \Hinf$ thus $G(s_0)$ is well defined and finite, therefore
    \[
    \lim_{s\to s_0} (s-s_0)^{m}f_i(s)=\lim_{s\to s_0}\sum_{j=1}^n \eta_jg_{ij}(s)=\sum_{j=1}^n \eta_jg_{ij}(s_0)=0
    \]
    hence $\eta \in \ker G(s_0)$.
\end{IEEEproof}
Note that the kernel condition is necessary and not sufficient unless $m=1$. The following result provides a Laplace-domain criterion to analyze the possible poles of $R(sI-A_0)^{-1}$ in terms of the coprime factorization of the closed loop system. Since any system that is stabilizable by feedback admits a coprime factorization over $\Hinf$ \cite{GSm:93}, assuming their existence is not a limiting assumption. A brief overview and results on the subject is available in Appendix \ref{sec:coprf}.
\begin{theorem}\label{thm:tfones}
Consider a system whose initial condition response can be described by \eqref{eq:init_laplace}, where $G$ has full normal rank and admits a left coprime factorization $G=\Mtil^{-1}\Ntil$. Assume that \eqref{eq:init_laplace}  reaches asymptotic agreement, a pole $\lambda\in \CRHP$ of $G(s)$ is also a pole of $R(sI-A_0)^{-1}$ only if there exists some $\eta \neq 0$ such that
\[
    %\ones \otimes \eta \in \im \adjug G(s_0).
    \ones \otimes \eta \in \ker \Mtil(\lambda).
\]
Moreover, if $G$ is square and $\lambda$ is not a zero of $G$, then $\ker \Mtil(\lambda)=\ker G^{-1}(\lambda)$.
\end{theorem}

\begin{IEEEproof}
Consider a left coprime factorization of $G=\Mtil^{-1}\Ntil$, since the system reaches agreement we have
        \[
        \Ntil x_0 = \Mtil (\ones \otimes (R(sI-A_0)^{-1}q_0))+\Mtil Q
        \]
        which is stable since $\Ntil\in \Hinf$. Since $\Mtil,Q\in \Hinf$ this implies that $\Mtil$ cancels all the unstable poles of $(\ones \otimes (R(sI-A_0)^{-1}q_0))$. 
                % expanding this expression we have
        % %
        % \[
        % \Mtil(s)(\ones \otimes (R(sI-A_0)^{-1}q_0)) = \Mtil(s)(\ones \otimes I_n) R(sI-A_0)^{-1}q_0.
        % \]
        % %
        Expanding the partial fraction decomposition of $R(sI-A_0)^{-1}$ yields
        \[
        R(sI-A_0)^{-1} = \sum \frac{1}{(s-s_i)^{m_i}}K_i
        \]
        where $m_i$ is the multiplicity of the pole and $K_i$ is the residue matrix and denote $\eta_i \coloneqq K_iq_0$. Rewriting the expression we have
       \[
        \Ntil x_0-\Mtil Q = \Mtil\sum_{i=1}(\ones \otimes \eta_i)\frac{1}{(s-s_i)^{m_i}},
        \]
        and from linearity we have
        \[
        \Ntil x_0-\Mtil Q\in \Hinf \iff \Mtil(\ones \otimes \eta_i)\frac{1}{(s-s_i)^{m_i}}
        \in \Hinf \quad \forall i.
        \]
        Recall that $\Mtil\in \Hinf$ and is square, thus we can apply Lemma \ref{lem:cancel} term by term and obtain the required result. Now, if $G$ is square, with full normal rank, and $\lambda$ is not a zero of $G$, by Lemma \ref{lem:coprime} (in the appendix) we get $\ker \Mtil(\lambda)=\ker G^{-1}(\lambda)$.
\end{IEEEproof}
Note that the existence of a coprime factorization is not a restrictive condition. In fact, a system admits a coprime factorization over $\Hinf$ if and only if it can be stabilized via feedback by some casual and LTI controller \cite{Sm:89}. If $G$ is a real-rational transfer function, then there are even efficient analytical ways to construct coprime factorizations, but this may not be the case for irrational transfer functions. However, for a large class of relevant systems this can be bypassed since the assumption that $G$ is square and has no intersection between right half-plane poles and zeros is satisfied. This is illustrated in the sequel, where we provide an in-depth analysis of such a class of systems.

It is important to note that Theorem \ref{thm:tfones} assumes that the system \emph{reaches} agreement, it does not guarantee it. The actual trajectory depends on the whole system and the structure of initial conditions, as illustrated in the following two examples.
\begin{example}\label{examp:nocontrol}
Consider the aggregate dynamics given by
\[
 G(s)=(s I_{n\nu}-I_\nu \otimes M )^{-1},
\]
for some matrix $M$ with eigenvalues on the imaginary axis. Then, we trivially have 
\[
G^{-1}(s_0)(\ones \otimes \eta)=0
\]
whenever $(s_0I-M)\eta=0$. This is because for certain initial conditions, the agreement space is in the image of $G(s)$. Indeed, since the agreement space is forward invariant, if the initial conditions are chosen to be the same for all agents, i.e. $x_0=\ones \otimes \tilde{x}_0$, clearly the agents would synchronize. \Endofremark
\end{example}
A different example concerns possible cancellations due to the initial conditions, resulting in certain poles being canceled on the agreement space.
\begin{example}\label{examp:zeros}
    Consider the system
    \[
    G(s)=(sI_4 -I_2\otimes \mmatrix{0&1\\-1&0}+\mmatrix{1&-1\\-1&1}\otimes I_2)^{-1},
    \]
    comprised of two harmonics oscillators with frequency $\omega_0=1$ and controlled by a static feedback gain. It is routine to verify that the system reaches agreement for all initial conditions, since the feedback gain corresponds to the Laplacian of the complete graph. It is easily seen that
    \[
    G^{-1}(\J)\left(\ones \otimes \mmatrix{-\J\\1}\right)=0, \quad G(-\J)\left(\ones \otimes \mmatrix{\J\\1}\right)=0,
    \]
    so the conditions of Theorem \ref{thm:tfones} are satisfied for $\lambda=\pm \J$ as expected. However, the trajectory itself may depend on the initial conditions. Consider for $x_0=\mmatrix{1&-1&-1&1}'$, and the resulting transfer function is
    %
    %\begin{multline*}
    \[
    G(s)x_0=\frac{s^2+1}{(s^2+1)(s^2+4s+5)}\mmatrix{s+1\\-(s+3)\\-(s+1)\\s+3} ,
    %\adjug G(s)x_0=(s^2+1)\mmatrix{s+1&-(s+3)&-(s+1)&s+3}'
    \]
    %\end{multline*}
    %
    thus for these initial conditions the marginally stable poles are canceled and the system's output will decay to zero, i.e. a trivial agreement trajectory. \Endofremark
\end{example}
These examples illustrate that Theorem \ref{thm:tfones} can only be used to predict possible components of an agreement trajectory, and not whether the trajectory itself always exists or is comprised of all the possible components. 

%%%%%%%%%%%%%%%%%%%%%%%%%%%%
\section{Agreement trajectories in homogeneous diffusively coupled networks} \label{sec:diffusive}
%%%%%%%%%%%%%%%%%%%%%%%%%%%%
Theorem \ref{thm:tfones} provides a necessary conditions for whether an unstable pole of a transfer function $G$ can contribute to a possible agreement trajectory by examining the null space of its left denominator $\Mtil$. The condition can be simplified if $G$ is square and well defined at the prospective unstable pole $\lambda$, in which case we can consider the null space of the inverse transfer function. In this section we will show how this relaxation can be used to easily analyze the behavior of standard diffusively coupled, homogeneous networks with full state information. Full state information implies that $B_i=I$, and we also assume without loss of generality that the agents are trying to achieve full state agreement, i.e. $C_i=I$.

Following the discussion in Section \ref{sec:intro}, we wish to analyze the agreement, viz. synchronization, behavior of model \eqref{eq:dynamics_agg} when subject to dynamic perturbations with an emphasis on delays. In all cases we assume that the perturbations are LTI, which allow for analysis in the Laplace domain. To facilitate the analysis via the Laplace transform, note that the aggregate version of the dynamics is given by
\begin{equation} \label{eq:dynamics_agg}
\begin{aligned}
    \dot{x}(t)&=(I_\nu \otimes A)x(t)\\&-\int_0^{t}\left((\Deg(t-\tau)-\Adj(t-\tau))\otimes I_n\right)x(\tau)d\tau\\
    &=(I_\nu \otimes A)x(t)-\int_0^t(\Lap(t-\tau) \otimes I_n)x(\tau)d\tau.
\end{aligned}  
\end{equation}
 Note that if we pick the dynamic kernel of these matrices to be an LTI system with state space realization $(0,0,0,\Lap)$ then we get the standard static feedback $u(t)=-(\Lap\otimes I) x(t)$. We refer to the above structured kernel as the \emph{dynamic Laplacian}, a notion that has been previously used  to study the stability of interconnected systems \cite{GM:04acc} and the convergence of general consensus-seeking algorithms \cite{HT:13}. 
\begin{remark}[Homogeneous vs Heterogeneous] \label{rem:het}
    We would like to emphasize that there is no significant loss of generality in considering homogeneous agents, since a necessary condition for asymptotic agreement in heterogeneous networks is the existence of a common internal model whose state synchronizes \cite[Thm. 3]{WSA:11}. Since, as discussed in the Section \ref{sec:intro}, the agreement trajectory is often an a-priori design requirement, cooperative output regulation methods often implement local reference generators within the controller a la \cite[Eq. 10]{WSA:11}. These reference generators are homogeneous and decoupled from the rest of the dynamics and follow dynamics such as \eqref{eq:dynamics_agg}, thus the analysis in this section can simply be applied for this part only. \Endofremark
\end{remark}

Assuming that the dynamic Laplacian admits a Laplace transform, the initial condition response of the closed loop system follows \eqref{eq:init_laplace} with transfer function
\begin{equation}\label{eq:dynLap}
 G(s)\coloneqq \left(I_\nu \otimes (sI_n-A) +\Lap(s)\otimes I_n\right)^{-1}
\end{equation}
where $\Lap(s)$ is the Laplace transform of the dynamic Laplacian. Note that \eqref{eq:dynLap} is square and has full normal rank, therefore for any $\lambda \in \CRHP$ that is not a zero of $G$ we have $\ker \Mtil(\lambda) =\ker G^{-1}(\lambda)$ which has a remarkably simple structure. This structure can be exploited to provide a tractable way to check whether $\lambda\in \CRHP$ can contribute to an agreement trajectory in the sense of Definition \ref{def:freqAg}.
\begin{theorem} \label{thm:dynLap}
Consider system \eqref{eq:dynamics_agg} with some dynamic Laplacian denoted by $\Lap(s)$ such that $\Lap\in \Hinf$ and its zeros do not accumulate on the imaginary axis, and assume that $G$ admits a coprime factorization over $\Hinf$. If \eqref{eq:dynamics_agg} reaches agreement, a pole $\lambda\in \CRHP$ can contribute to the agreement trajectory only if 
\[
\Lap(\lambda)\ones=\alpha(\lambda)\ones \quad \text{and} \quad \lambda+\alpha(\lambda)\in \spec A
\]
for some possibly complex scalar $\alpha(\lambda)$.
\end{theorem}
\begin{IEEEproof}
    Consider $G(s)$ as defined in \eqref{eq:dynLap}. By Theorem \ref{thm:tfones} $\lambda\in \CRHP$ can contribute to the agreement trajectory only if there is some nonzero vector $\eta$ such that $\ones \otimes \eta \in \ker \Mtil(\lambda)$, where $\Mtil$ is any left denominator of $G$. Since $sI-A$ is polynomial and $\Lap\in \Hinf$,  $G^{-1}$ is holomorphic in the open and bounded in the closed right half-plane. By definition the zeros of $G$ are the singularities of $G^{-1}$, thus $G$ has no right-half plane zeros at all and for all $\lambda \in \CRHP$ we have $\ker \Mtil(\lambda)=\ker G^{-1}(\lambda)$.

    Let $\lambda \in \CRHP$ and assume that such $\eta$ exists, then we have
    \[
    G^{-1}(\lambda)(\ones\otimes \eta)=\ones \otimes ((\lambda I_n-A)\eta)+(\Lap(\lambda)\ones)\otimes \eta =0.
    \]
    Denoting $q(\lambda)\coloneqq\Lap(\lambda)\ones$ and rewriting the Kronecker product explicitly yields
    \[
     \begin{bmatrix}
        (\lambda I_n -A)\eta\\ \vdots \\ ( \lambda I_n -A)\eta
    \end{bmatrix} = - \begin{bmatrix}
        q_1(\lambda)\eta \\ \vdots \\ q_\nu(\lambda) \eta
    \end{bmatrix}
    \]
    where $q_i(\lambda)$ is the $i$th component of $q(\lambda)$. Clearly a necessary condition for equality to hold is that 
    \[
    \alpha(\lambda) \coloneqq q_1(\lambda)=q_2(\lambda)=\cdots =q_\nu(\lambda)
    \]
    for some possibly complex scalar $\alpha(\lambda)$. The result follows immediately by rewriting the above using $\alpha(\lambda)$.
\end{IEEEproof}
\begin{remark} \label{rmk:feedback}
    Note that it is  a trivial matter to extend Theorem \ref{thm:dynLap} to include two common cases: i) a state feedback gain $K$, and ii) a dynamic observer. In the first case $G^{-1}$ will include the term $\Lap(s)\otimes K$ and the resulting condition will simply be
    \[
    (\lambda I_n-(A-\alpha(\lambda)K))\eta=0 \implies \lambda\in \spec \left(A-\alpha(\lambda)K\right).
    \]
    This case is technically more difficult to verify, but is not conceptually different from the case of $K=I_n$. The second case follows directly from the first, since such systems can be always be put into the previous structure for some augmented matrices, for instance \cite{LDCH:10,TTM:13}.\Endofremark
\end{remark}
Theorem \ref{thm:dynLap} provides general necessary conditions for some mode/frequency to contribute to an agreement trajectory. It does not guarantee that agreement will be reached, nor that all possible modes will indeed contribute to the trajectory even if it is reached. Therefore Theorem \ref{thm:dynLap} is suitable for \emph{analysis} of protocols rather than design. However, since the required trajectory is often determined a-priori by the eigenvalues of $A$, it is a useful tool for examining the robustness of the trajectory. The following result is an immediate corollary of Theorem \ref{thm:dynLap} in this vein.
\begin{corollary}\label{cor:retain}
Under the conditions of Theorem \ref{thm:dynLap}, the nominal agreement trajectory can be preserved only if
   \begin{equation} \label{eq:robust1}
     \Lap(\lambda)\ones =0, \quad \forall \lambda \in \spec A \cap \CRHP.  
   \end{equation}
   %
  % If \eqref{eq:robust} holds, then we say that the nominal trajectory is \emph{robust} with respect to the dynamic Laplacian $\Lap(s)$.
\end{corollary}
Corollary \ref{cor:retain} provides an explicit necessary condition for a nominal agreement trajectory to be robust to perturbations. This condition has a straightforward interpretation - the Laplacian dynamics must not interfere with the nominal unstable poles. This is naturally the case for the static Laplacian because by design $\ones \in \ker \Lap$, thus if $\Lap(\lambda)=\Lap$ then $\lambda$ is not effected by the Laplacian dynamics. A special case of condition \eqref{eq:robust1} for integrator agents with a distributed delay appears in \cite[Prop. 2.8]{MMN:09}, derived from a geometric perspective and in state-space. 

In contrast, by Theorem \ref{thm:dynLap} it is clear that in some cases the Laplacian dynamics may result in a \emph{different} agreement trajectory. This discussion motivates a more formal discussion on robustness and fragility of diffusively coupled systems from the perspective of the resulting trajectories. This is the focus of the following section, where we explicitly analyze several common cases.
%%%%%%%%%%%%%%%%
\section{Fragility and robustness of trajectories} \label{sec:frag}
%%%%%%%%%%%%%%%%
 
The first step in our analysis is to properly define robustness of agreement trajectories. 
\begin{definition}[Robust agreement trajectories] \label{def:robust}
Consider system \eqref{eq:dynamics_agg} with transfer function \eqref{eq:dynLap}, where $\Lap(s)$ and $G(s)$ satisfy the conditions of Theorem \ref{thm:dynLap}. We say $A$ generates a \emph{robust trajectory} with respect to $\Lap(s)$ if it can be preserved in presence of perturbations and if
    \begin{equation}\label{eq:robust2}
        \Lap(\lambda)\ones=\alpha(\lambda)\ones \implies \lambda \in \spec A \cap \CRHP.   
    \end{equation}
If the trajectory is not robust, then we say that it is \emph{fragile}.
\end{definition}
Robustness in the sense of Definition \ref{def:robust} implies that not only can the nominal trajectory be preserved, but that the unstable eigenvalues of $A$ are the only poles capable of generating an agreement trajectory. Moreover, robustness in this sense is, in general, dependent on both $A$ and $\Lap(s)$, with some trajectories being robust to certain network dynamics while others being fragile. In the rest of this section, we shall focus on four common types of network dynamics grouped into two distinct groups: structure preserving dynamics and transmission dynamics.

%%%%%%%%%%%%%%%%%%%%%%%%%%%%%%
 \subsection{Structure preserving dynamics} \label{sec:main:preserve}
%%%%%%%%%%%%%%%%%%%%%%%%%%%%%
% Since this definition is independent of $A$, it is reasonable to assume that if it holds $\Lap(s)$ must have some special structure that is related to that of the static Laplacian.
%
\begin{definition}[Structure-preserving dynamics] \label{def:preserv}
    Let $\Lap(s)$ be the transfer function of a dynamic Laplacian. We say that the dynamics are \emph{structure-preserving} if they preserve the null space of the static Laplacian.
    %the image and null space of the static Laplacian. 
\end{definition}
It is clear that there are three possible forms for structure-preserving dynamics:
\begin{enumerate}
    \item network output dynamics $\Lap(s)=\Delta(s)\Lap$,
    \item network input dynamics $\Lap(s)=\Lap \Delta(s)$, and
    \item edge dynamics $\Lap(s)=E\diag\{\Delta_i(s)\} E'$ where $E$ is the incidence matrix.
\end{enumerate}
\begin{remark}
    Recall that if $\Gr$ is directed, then the static Laplacian can be written as $\Lap=BE'$ for some binary matrix $B$, slightly modifying the edge dynamics, but still making them structure preserving. \Endofremark
\end{remark}
Structure preserving dynamics are by far the most common ones considered in the literature. These encompass actuation delays \cite{O-SM:04,QGY:19}, weighted interconnections \cite{ZB:17}, and systems with additive uncertainty \cite{TTM:13}. These are not only limited to perturbed systems, since the third type covers any purely diffusively coupled system \cite{BMZ:23}. The following result formalizes the trajectory-robustness property of structure-preserving dynamics.
\begin{theorem}\label{thm:robustPreserve}
   If the conditions of Theorem \ref{thm:dynLap} are satisfied, the Laplacian dynamics are structure preserving, and the underlying graph is weakly connected, then every nominal trajectory is robust.
   %then only the eigenvalues of $A$ can contribute to an agreement trajectory.
\end{theorem}
\begin{IEEEproof}
    Recall that for a weakly connected graph $\Lap$ has a simple eigenvalue at $0$ and $\ker \Lap = \im \ones$ \cite[Prop 3.8]{ME:10}. Moreover, since $\Lap(s)\in \Hinf$ we must have that $\Delta \in \Hinf$ and in particular $\Delta(\lambda)$ is well defined. Thus, for the three structure-preserving dynamics as in Definition \ref{def:preserv} we have the following results.
    \begin{enumerate}
    %[label=\roman*]
        \item  Since $\Delta(\lambda)$ is well defined and $\Lap\ones=0$, we must have $\Lap(\lambda)\ones=0$.
        \item Assume by contradiction that $\Lap(\lambda)\ones=\alpha(\lambda)\ones$ for a non-zero $\alpha(\lambda)$, and let $\eta\neq 0$ denote the left eigenvector of $\Lap$ corresponding to eigenvalue $0$. Since by assumption $\Lap(\lambda)\ones \neq 0$ we have
        \[
        \eta' \Lap \Delta(\lambda)\ones =0 \iff \alpha(\lambda)\eta' \ones =0,
        \]
        but since $0$ is a simple eigenvalue $\eta' \ones \neq 0$, a contradiction.
        \item The final result follows by the same arguments as i. and the fact that by construction $\ones \in \ker E' $ for any graph, e.g. \cite[Def. 2.5]{ME:10}.
    \end{enumerate}
    Thus, for a structure preserving perturbation $\Lap(\lambda)\ones =\alpha(\lambda)\ones$ if and only if $\alpha(\lambda)=0$. Consequently, by Theorem \ref{thm:dynLap} any trajectory is robust as in Definition \ref{def:robust}.
\end{IEEEproof}
Before we move on to other types of dynamic Laplacians, it is important to notice that robustness of the original trajectory does not imply that we reach agreement. The only guarantee is that if the systems reaches agreement on a non trivial trajectory, then this trajectory must be generated by the eigenvalues of $A$. There are numerous works which consider structure-preserving network structures that require additional constraints on the parameters to realize the agreement trajectory, e.g. \cite{O-SM:04,MMN:09,ZB:17,QGY:19}.

%%%%%%%%%%%%%%%%%%%%%%%%%%%%%%
 \subsection{Transmission dynamics} \label{sec:main:transdyn}
%%%%%%%%%%%%%%%%%%%%%%%%%%%%%
The second common type of network dynamics are transmission only dynamics, which are defined as follows.
\begin{definition}[Transmission dynamics] \label{def:trans}
    Let $\Lap(s)$ be the transfer function of a dynamic Laplacian. We say that the Laplacian has are \emph{transmission only dynamics} if only the adjacency matrix is dynamic, i.e. $\Lap(s)=\Deg-\Adj(s)$. In this case we denote by $\Delta_{ij}(s)$ the $ij$th entry of $\Adj(s)$.
    %the image and null space of the static Laplacian. 
\end{definition}
% At:14
The most common transmission dynamics are transmission delays, which have been widely studied in the context of agreement to a constant, i.e. consensus \cite{M:05,SDJ:08,MPA:10}. The following result is a variation of Theorem \ref{thm:dynLap} for this special case.
\begin{theorem}\label{thm:generaltrans}
   Assume that the conditions of Theorem \ref{thm:dynLap} are satisfied and the Laplacian dynamics are transmission only. If \eqref{eq:dynamics_agg} reaches agreement, a pole $\lambda\in \CRHP$ can contribute to the agreement trajectory only if 
   \[
     d_i-d_k=\sum_{j=1}^{\nu}(\Delta_{ij}(\lambda)-\Delta_{kj}(\lambda))  \quad \forall i,k
    \]
    where $d_i$ is the degree of the $i$th agent and
    \[
    \lambda+d_i-\sum_{j=1}^{\nu}\Delta_{ij}(\lambda)\in \spec A.
    \]
    Equivalently, we require that
    \[
    (\Adj-\Adj(\lambda))\ones=\alpha(\lambda)\ones \quad \text{and} \quad \lambda+\alpha(\lambda)\in \spec A
    \]
    where $\Adj$ is the static adjacency matrix.
% %
% \[
% \Lap(\lambda)\ones=\alpha(\lambda)\ones \quad \text{and} \quad \lambda+\alpha(\lambda)\in \spec A
% \]
% %
% for some possibly complex scalar $\alpha(\lambda)$.
%    %then only the eigenvalues of $A$ can contribute to an agreement trajectory.
\end{theorem}
\begin{IEEEproof}
    Consider a transmission only dynamic Laplacian $\Lap(s)=\Deg-\Adj(s)$. Applying Theorem \ref{thm:dynLap} we can obtain
    \[
    \Lap(\lambda)\ones =\alpha(\lambda)\ones =\begin{bmatrix}
        d_1-\sum_{j=1}^{\nu}\Delta_{1j}(\lambda) \\ \vdots \\ d_\nu-\sum_{j=1}^{\nu}\Delta_{\nu j}(\lambda).
    \end{bmatrix}
    \]
    Since $\alpha(\lambda)$ is a scalar all entries of the RHS must be identical, thus
    \[
    \alpha(\lambda) = d_i-\sum_{j=1}^{\nu}\Delta_{ij}(\lambda) \quad \forall i,
    \]
    from which the second condition follows immediately. Choosing arbitrary $i$ and $k$, equating them, and rearranging yields the first condition. The equivalent characterization follows from the fact that every digraph satisfies $\Deg\ones=\Adj\ones$ \cite[\S 2.3.5]{ME:10}.
\end{IEEEproof}
Comparing Theorems \ref{thm:robustPreserve} and \ref{thm:generaltrans}, it is evident that the more the dynamics affect the null space of the Laplacian, the less likely it is to reach agreement. Both characterizations in Theorem \ref{thm:generaltrans} indicate that to analyze the agreement trajectory we must analyze every pair of agents in the graph at all possible $\lambda$. In particular, in order to preserve the original trajectory it is clear that we must have $\Adj(\lambda)=\Adj$ for all $\lambda\in \CRHP \cap \spec A$. Essentially, the dynamics must leave the original poles \emph{invariant}. This insight is the key idea as to why consensus is remarkably robust against constant \cite{SDJ:08}, distributed \cite{At:14}, and even heterogeneous \cite{MPA:10} transmission delays, as formalized in the following corollary.
\begin{corollary} \label{cor:consDelay}
    Consider system \eqref{eq:dynamics_agg} with constant, heterogeneous, transmission only delays. If $0\in \spec A$, then it can always contribute to an agreement trajectory.
\end{corollary}
\begin{IEEEproof}
    Let $\tau_{ij}$ denote the transmission delay between agent $i$ and $j$, then $\Delta_{ij}(s)=\E^{-\tau_{ij}s}$. Clearly $\Delta_{ij}(0)=1$ for all $\tau_{ij}$, therefore $\Adj(0)=\Adj$ and the result follows immediately from Theorem \ref{thm:generaltrans}.
\end{IEEEproof}
If we consider only uniform transmission dynamics, i.e. $\Adj(s)=\Delta(s)\Adj$ for some scalar transfer function $\Delta(s)$, we can obtain even simpler conditions.
\begin{corollary}\label{cor:UniTrans}
Assume that the conditions of Theorem \ref{thm:dynLap} are satisfied, and that $\Lap(s)=\Deg-\Delta(s)\Adj$ has uniform transmission dynamics. If \eqref{eq:dynamics_agg} reaches agreement, a pole $\lambda\in \CRHP$ can contribute to the agreement trajectory only if one of the following holds:
 \begin{enumerate}
        \item $\Delta(\lambda)=1$ and $\lambda\in \spec A$.
        \item $\Gr$ is $d-$regular and $\lambda+d(1-\Delta(\lambda))\in \spec A$.
\end{enumerate}
    % Assume that the Laplacian is effected by a uniform transmission perturbation, i.e. $\Lap(s)=\Deg-F(s)\Adj$. If \Ass{\ref{ass:M},\ref{ass:L}} hold, then a pole $s_0=\J\omega_0$ can contribute to an agreement trajectory only if one of the following holds:
    %  \begin{enumerate}
    %     \item $F(s_0)=1$ and $s_0\in \spec M$.
    %     %
    %     \item $\Gr$ is $d-$regular and $s_0+d(1-F(s_0))\in \spec M$.
    % \end{enumerate}
\end{corollary}
\begin{IEEEproof}
    Recall again that $\Deg\ones=\Adj\ones$, therefore
    \[
    \Lap(\lambda)\ones=(1-\Delta(\lambda))\begin{bmatrix}
        d_1 \\ \vdots \\ d_\nu
    \end{bmatrix}
    \]
    where $d_i$ are the elements of $\Deg$. By Theorem \ref{thm:dynLap} the above needs to be in the image of the all-ones vector, thus there are two possible cases. If $\alpha(\lambda)=0$, then $\Delta(\lambda)=1$ and it follows that $\lambda\in \spec A$ for the equality to hold.  If $\alpha(\lambda)\neq 0$, then $d_1=\cdots=d_\nu \coloneqq d$ implying that $\Gr$ must be $d$-regular.  In this case $\alpha(\lambda)=d(1-\Delta(\lambda))$ and the condition follows from Theorem \ref{thm:dynLap}.
\end{IEEEproof}
Notably, the well known robustness of consensus to transmission delays is not an intrinsic property of the consensus protocol's structure, but rather a consequence of well known properties of the delay operator. Namely, the delay operator has unit gain for all frequencies, and a linearly decreasing phase given by $\measuredangle \E^{-\tau \J \omega}=-\tau \omega$. Note that the first condition in Corollary \ref{cor:UniTrans} can be written in terms of \emph{both} the gain and phase of the perturbation 
\[
% \Delta(\J\omega_0)=1 \iff \abs{\Delta(\J\omega_0)}=1 \quad \text{and} \quad \measuredangle \Delta(\J\omega_0)=2\pi k.
\Delta(\J\omega_0)=1 \iff \abs{\Delta(\J\omega_0)}=1 \quad \text{and} \quad \measuredangle \Delta(\J\omega_0)=0.
\]
%
%\todo[inline]{Is it $2\pi k$ or identically zero?}
The delay element satisfies the phase condition precisely when $\tau \omega = 2\pi n$, but only for $\omega=0$ the frequency is delay independent. Therefore this robustness is due to the special behavior of the delay element at zero frequency, and in general is not robust for other stable perturbations. An example illustrating this, and in particular the importance of the phase condition, is provided in \S \ref{sec:ex:int}. Moreover, using Corollary \ref{cor:UniTrans} we can show that in general harmonic agreement trajectories are fragile for \emph{arbitrary small} transmission delays.
\begin{corollary} \label{cor:Uni_Trans_Del}
Assume that the conditions of Theorem \ref{thm:dynLap} are satisfied, that $\Lap(s)$ has some uniform transmission delay $\tau>0$, and that $\spec A=\pm \J\omega_1,\ldots, \pm\J \omega_q $ with $\omega_i\neq 0$. If 
    \[
    \tau \neq \frac{2\pi n}{\omega_i} , \quad n\in \mathbb{Z}
    \]
    for all $i$, then the system does not synchronize. In particular, the agents will not synchronize to a non-constant trajectory for any delay satisfying
    \[
    \tau < \frac{2\pi}{\max\{\omega_i\}}.
    \]
\end{corollary}
\begin{IEEEproof}
    Consider the two necessary conditions in Corollary \ref{cor:UniTrans} with $\Delta(s)=\E^{-\tau s}$. For the first condition, we have $\lambda=\J\omega_i\in \spec A$ and
    \[
    \E^{-\J\tau\omega_i}=1 \implies \tau=\frac{2\pi n}{\omega_i} \quad n\in \mathbb{Z},
    \]
    while the second condition requires $\Gr$ to be $d$-regular and
    \[
    \J\omega_i+d(1- \E^{-\J\tau\omega_i})\in \spec A.
    \]
    By assumption the first condition does not hold, so let us examine the second. Applying Euler's formula, we can simplify the expression to
    \[
    d(1-\cos(\omega_i\tau))+\J (\omega_i+d\sin(\omega_i \tau))\in \spec A.
    \]
    Note that since $\Gr$ is $d$-regular, $d$ is strictly positive and that $0\leq 1-\cos(\omega_i \tau)\leq 2$, thus the real part of this pole is non-negative. By assumption $A$ does not have poles with a positive real part, so the only possibility is if the real part is zero. Assume that it is, then
    \[
    1-\cos(\omega_i\tau)=0 \implies \omega_i\tau = 2\pi n, \quad \in \mathbb{Z}
    \]
    which reduces back to the first case which does not hold by assumption. Finally, since $\tau,\omega_i>0$ and $n$ is an integer
    \[
    \tau < \frac{2\pi}{\max\{\omega_i\}} \implies \tau \neq \frac{2\pi n}{\omega_i} \quad \forall n\in \mathbb{Z}, \, \omega_i \in \spec A
    \]
    therefore the agents will never synchronize to a non-constant trajectory.
\end{IEEEproof}
It is valuable to note the essential difference between the consensus case and the case of general synchronous trajectories. Namely the consensus trajectory is robust to arbitrarily large transmission delays, while general periodic trajectories are \emph{fragile} for arbitrarily small transmission delays. Such delays are almost inevitable in practical applications where the agents communicate information, in particular when using output regulation techniques where the agents exchange internal controller variables \cite{WSA:11}. Consequently, output regulation approaches aiming to synchronize with periodic trajectories must always be augmented in someway to account for possible small transmission delays. This can be accomplished in a relatively straight-forward manner via distributed predictors if the delays are known \cite{LL:17know} or measured \cite{BMZ:24cdc}, but may pose a significant challenge when the delays are unknown. One potential way to circumvent this issue is to exploit the insight that consensus as an objective is robust to arbitrary transmission delays. By having the agents communicate and agree on some dummy variable $\zeta_i(t)$ and implementing local reference signals such as $r_i(t)=\E^{A_0 t}\zeta_i(t)$. If $\zeta_i(t)$ reach agreement with a constant, then asymptotically 
\[
\lim_{t\to \infty} \norm{r_i-r_j}=0 \quad \forall i,j
\]
and all that remains is to ensure local tracking of $r_i$ 
à la \eqref{eq:timeagreement:sync}. This is essentially the mechanism in \cite[\S-3.2
]{LL:17unk}, and the same idea can potentially extend to emulation-based asynchronous sampled-data controllers such as \cite{BMZ:23cdc,BMZ:24ECC}.
\begin{remark}[Agents with nominal poles in $\RHP$] \label{rem:transunstable}
    It is important to remark that Corollary \ref{cor:Uni_Trans_Del} assumes that all the poles of $A$ are on the imaginary axis, if this assumption fails then we must resort to the more general Corollary \ref{cor:UniTrans}. An example of this is fact is provided in \S \ref{sec:ex:frag}. \Endofremark
\end{remark}
%
%%%%%%%%%%%%%%%%%%%%%%%%%%%%%
\subsection{Discussion}\label{sec:main:disc}
%%%%%%%%%%%%%%%%%%%%%%%%%%%
It is instructive to contextualize these results within the broader framework of output regulation. It is well established that a necessary and sufficient condition for synchronization is the existence of a common internal model among all agents \cite{WSA:11}. This equivalence has led to numerous solutions to the agreement problem based on solving decoupled output regulation problems locally for each agent. Despite this equivalence, there remains a subtle practical difference between classical output regulation and its modern cooperative counterpart.

In standard output regulation problems, an explicitly known exosystem generates the reference signal to track and the disturbance signal to reject. This exosystem is incorporated into the design via the celebrated internal model principle \cite{FrW:76}, yielding a stabilizing servo controller. From a frequency-domain perspective, regulation is achieved by placing closed-loop zeros in the transfer function mapping the input to the regulation error. Hence, the controller must perfectly incorporate the model of the exosystem, and even a small mismatch in that model can lead to unbounded tracking errors \cite[Thm. 3]{DG:75}. When the reference signal is perfectly known, this requirement is generally unproblematic, as the local controller implementation is accurate up to numerical precision. Nevertheless, this stringent requirement is the fundamental reason why closed-loop stability is considered a robust property, whereas output regulation is not \cite[Thm. 4]{Fr:77}.

A stark difference in the agreement problem is that no such exosystem exists in the traditional sense, as the system is driven solely by initial conditions. Therefore, the controller places the internal model's \emph{poles} in the closed loop to generate a non-trivial trajectory. This not only compromises internal stability \cite{BMZ:23} but also implies that the ``exosystem'' is \emph{internal}. Because it is realized within the loop and across the network, it is highly susceptible to unmodeled network dynamics. As a result, its poles---and by extension, the agreement trajectory---may shift. As discussed in \S~\ref{sec:agrDef}, it is common to define agreement through \eqref{eq:timeagreement:out} without explicitly specifying the target trajectory. Thus, even under architectures designed for robust agreement, network perturbations may yield significantly altered, or even unbounded, agreement trajectories.

\begin{remark}[Generalizations]
    We emphasize the generality of Theorem \ref{thm:tfones}, which requires only the existence of a coprime factorization and a finite-dimensional agreement trajectory. The subsequent focus on Laplacian-based systems, such as \eqref{eq:dynLap}, reflects their structural simplicity and prevalence in the literature rather than any inherent limitation of our framework. Similar analyses can be applied to many other LTI agreeing systems, including those utilizing the adjacency matrix \cite{Lee:16}, dynamic protocols \cite{TTM:13}, or two-degrees-of-freedom architectures \cite{BMZ:25}. Furthermore, the preceding discussion holds for these generalized cases almost verbatim. \Endofremark   
\end{remark}
 \vspace{-.3cm}
%%%%%%%%%%%%%%%%%%%%%%%%%%%%
\section{Numerical examples} \label{sec:ex}
%%%%%%%%%%%%%%%%%%%%%%%%%%%%
Below are several numerical examples illustrating the various results of this paper for the simple case of $\nu=3$ nodes and the two graphs shown in Fig.~\ref{fig:ExG}. Although we are considering a very simple graph structure, the numerical illustrations can be easily extended to larger graphs. We preferred to keep the graph simple because these small size graphs are sufficient to illustrate our theoretical results. The static degree and adjacency matrices for both graphs are given by
\[
\begin{aligned}
    \Deg^{(1)}&=\diag\{2,2,2\} & & \Adj^{(1)}=\mmatrix{0&1&1\\1&0&1\\1&1&0}\\
    \Deg^{(2)}&=\diag\{2,1,1\} & & \Adj^{(2)}=\mmatrix{0&1&1\\1&0&0\\0&1&0},
\end{aligned}
\]
note that the complete undirected graph is in particular a $2$-regular graph. The first set of examples illustrate the transmission delay fragility of synchronizing systems, while the next set illustrate how the agents may converge to a different trajectory due to transmission dynamics. 
\begin{figure}[!hbt]
 \centering
   \subfigure[Complete undirected graph.]{\label{fig:Ex:g1}\includegraphics[width=0.32\columnwidth,clip]{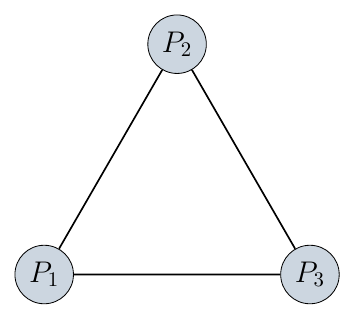}} \hspace{1cm} 
   \subfigure[A weakly connected digraph.]{\label{fig:Ex:g2}\includegraphics[width=0.32\columnwidth,clip]{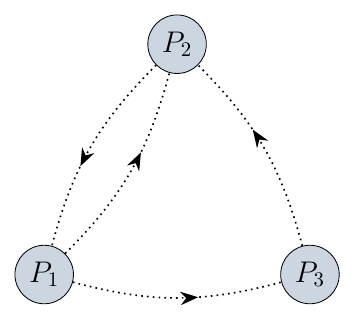}}
   %\hfill 
   % \subfigure[General transmission.]{\label{fig:pert:transge}\includegraphics[width=0.3\columnwidth,clip]{figs/Trans_gent.pdf}}
 \caption{The two graphs in the examples.} \label{fig:ExG}
 \vspace{-.5cm}
\end{figure}
%
%%%%%%%%%%%%%
\subsection{Fragility of synchronization to transmission delays} \label{sec:ex:frag} 
%%%%%%%%%%%%%
To illustrate the fragility for transmission delays, consider a nominal system interacting over the network in Fig.~\ref{fig:Ex:g1} and set to synchronize with a sine wave with frequency $\omega_0=1$, corresponding to
\begin{subequations} \label{eq:ex1}
\begin{equation} \label{eq:ex1:1}
   A=\begin{bmatrix}
    0&1\\-1&0
\end{bmatrix}. 
\end{equation}
%
% \[
% A=\begin{bmatrix}
%     0&1\\-1&0
% \end{bmatrix}.
% \]
Since $\omega_0=1$, by Corollary \ref{cor:Uni_Trans_Del} the agents cannot synchronize for any delay $\tau\neq 2\pi n$ for some integer $n$. Fig. \ref{fig:Ex1} consists of three simulations for delay values of $\tau=0.1,0.9,2\pi$, the nominal trajectories are superimposed in dashed lines. 
%width=0.9\columnwidth
\begin{figure}[!hbt]
 \centering
   \subfigure[Trajectories with $\tau=0.9$.]{\label{fig:Ex1:1}\includegraphics[scale=0.38,clip]{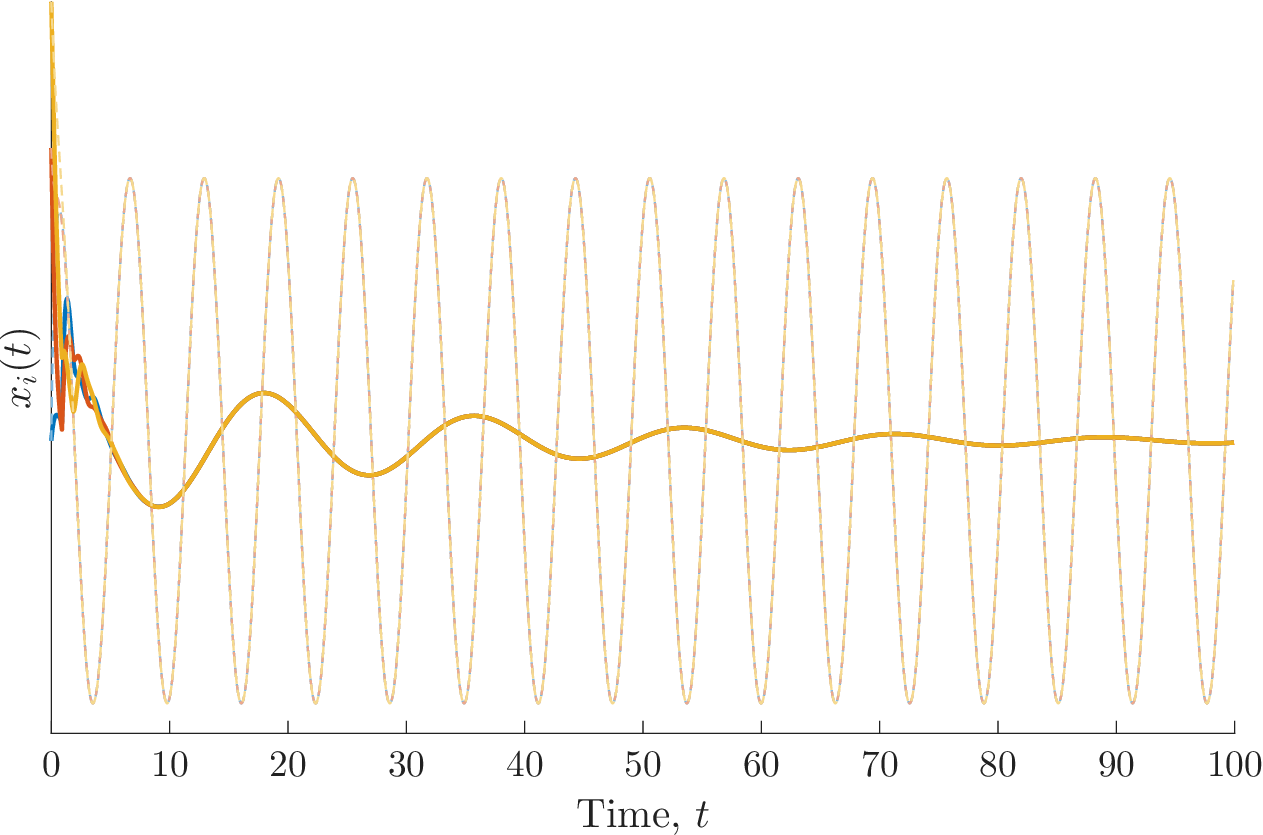}}\\
   \subfigure[Trajectories with $\tau=0.1$.]{\label{fig:Ex1:2}\includegraphics[scale=0.38,clip]{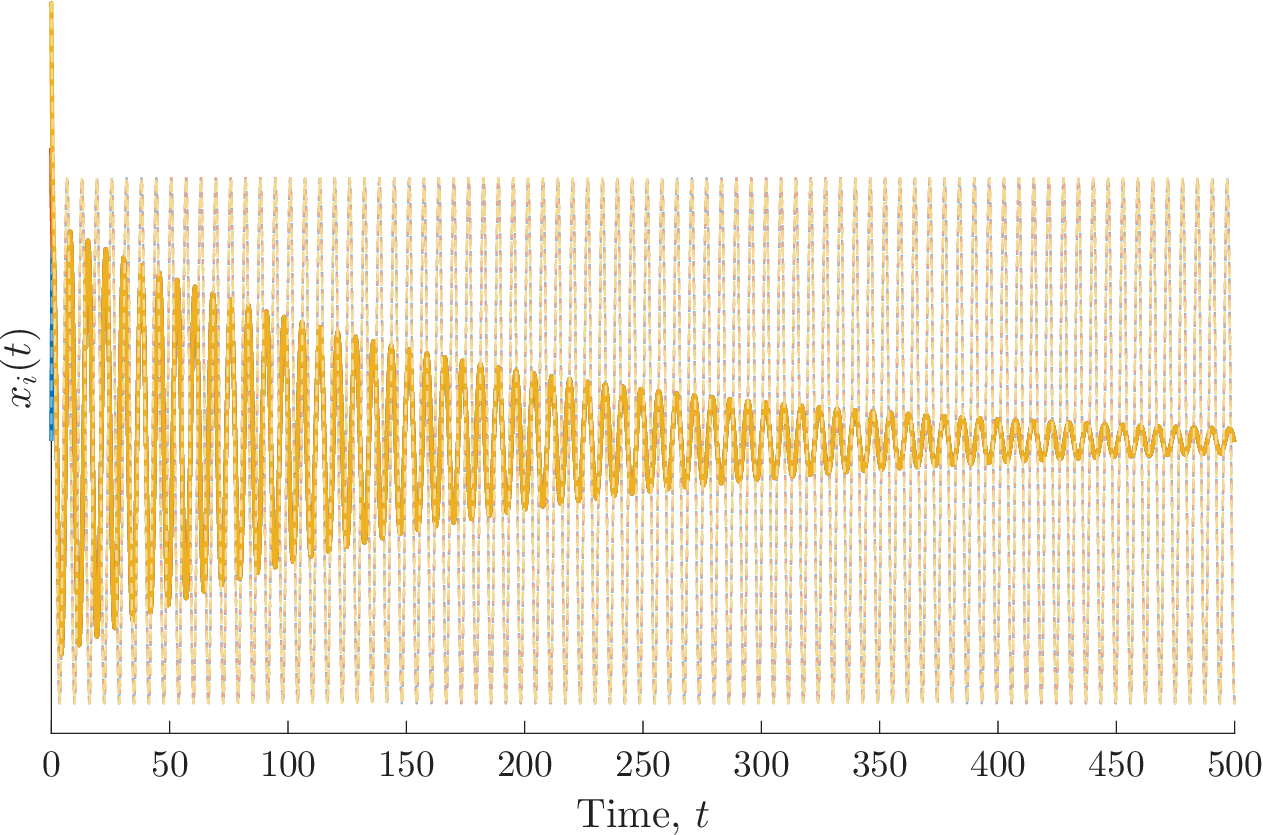}}\\
   \subfigure[Trajectories with $\tau=2\pi$.]{\label{fig:Ex1:3}\includegraphics[scale=0.38,clip]{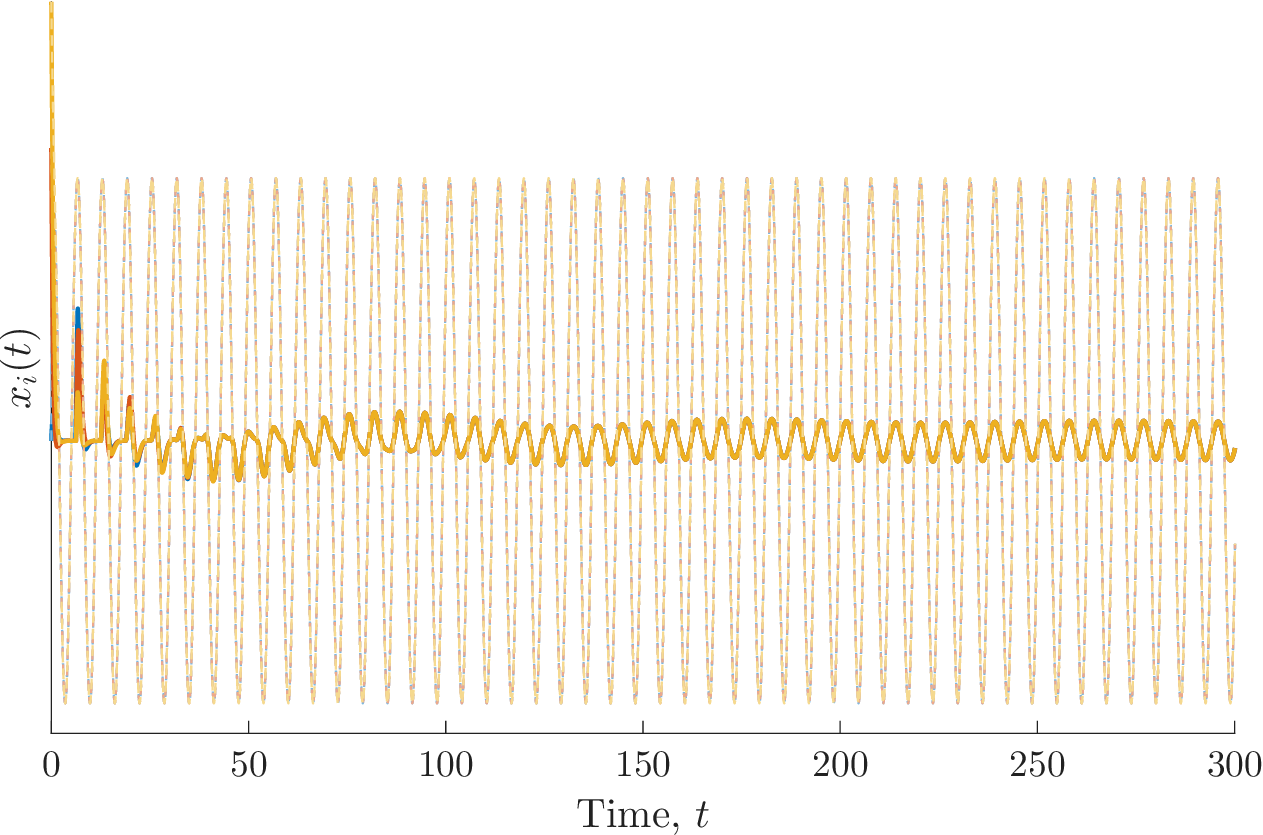}}
   %\hfill 
   % \subfigure[General transmission.]{\label{fig:pert:transge}\includegraphics[width=0.3\columnwidth,clip]{figs/Trans_gent.pdf}}
 \caption{Resulting trajectories for example \ref{sec:ex:frag} with \eqref{eq:ex1:2}. Delayed (solid) and nominal (dashed) trajectories.} \label{fig:Ex1}
 \vspace{-.2cm}
\end{figure}
It can be readily seen that indeed the agents do not asymptotically synchronize to a nontrivial trajectory for $\tau=0.1,0.9$, but do for $\tau=2\pi$. It bears to mention that the added delay changes the amplitude of the resulting trajectory, even if it does not effect the frequency.

%rem:transunstable
Finally, we illustrate the importance of the condition $\spec A \subset \J\mathbb{R}$ in Corollary \ref{cor:Uni_Trans_Del}. For simplicity assume that the some system is affected by a uniform delay of $\tau=1$, and there is some $\lambda \in \spec A$ with strictly positive real part. Applying the second condition of Corollary \ref{cor:UniTrans} for some arbitrary frequency $\omega_0$ implies that if
\[
\J\omega_0+2(1-\E^{-\J\omega_0})=\lambda \implies \lambda = 2(1-\cos(\omega_0))+\J (\omega_0+2\sin(\omega_0))
\]
holds, then the agents can synchronize to a harmonic trajectory at frequency $\omega_0$. However, since $0\leq 2(1-\cos(\omega_0)\leq 4$ this implies that the uncontrolled agents must have poles in the open right half-plane. Consider for example $\omega_0=\pi/3$, hence any matrix with poles at $\lambda=1\pm \J\left(\frac{\pi}{3}+\sqrt{3}\right)$ can potentially synchronize to a trajectory with frequency $\pi/3$. Figure \ref{fig:Ex2} illustrates exactly this setup with
\begin{equation} \label{eq:ex1:2}
  A=\mmatrix{1&\frac{\pi}{3}+\sqrt{3}\\-\left(\frac{\pi}{3}+\sqrt{3}\right) &1}, \quad \spec A =1\pm   \J\left(\frac{\pi}{3}+\sqrt{3}\right)
\end{equation}
\end{subequations}
% \[
% A=\mmatrix{1&\frac{\pi}{3}+\sqrt{3}\\-\left(\frac{\pi}{3}+\sqrt{3}\right) &1}
% \]
%
which indeed synchronize albeit after a long transient.
%width=0.9\columnwidth
\begin{figure}[!hbt]
 \centering
 \includegraphics[scale=0.75,clip]{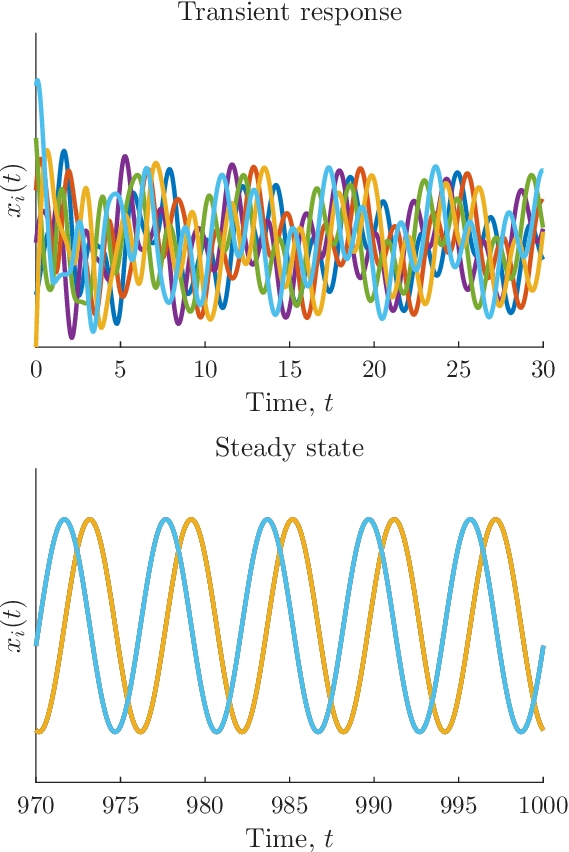}
 \caption{Resulting trajectories for example \ref{sec:ex:frag} with $\tau=1$ and \eqref{eq:ex1:2}.} \label{fig:Ex2}
 \vspace{-.2cm}
\end{figure}
%
%%%%%%%%%%%%%
\subsection{Change of agreement trajectories for $d$-regular graphs}\label{sec:ex:int}
%%%%%%%%%%%%%
%
Now we shall illustrate the second case in Corollary \ref{cor:UniTrans}. Consider simple integrator agents, i.e. $A=0$, and and let the transmission perturbation be
\begin{subequations}\label{eq:ex:trans}
    
\begin{equation}\label{eq:ex:trans:1}
   \Delta(s)=2\frac{s-1}{s+2}\implies \Lap(s)=\Deg-2\frac{s-1}{s+2} \Adj 
\end{equation}
% %
% \[
% \Delta(s)=2\frac{s-1}{s+2}\implies \Lap(s)=\Deg-2\frac{s-1}{s+2} \Adj
% \]
%
Note that $\abs{\Delta(0)}=1$ but $\measuredangle \Delta(0) =-\pi$, thus the gain condition discussed in \S~\ref{sec:main:transdyn} holds, but the phase condition does not. Consequently, the agents will not be able to preserve their original consensus trajectory when faced with this perturbation, but they might synchronize to a different trajectory if
\[
\J\omega+2(1-\Delta(\J\omega)) =0
\]
has a solution. It turns out that
\[
\J\omega+d(1-\Delta(\J\omega))=\frac{8-\omega^2}{\J\omega+2}
\]
therefore the agents can (and do) synchronize to a \emph{sine wave} with frequency $\omega=\sqrt{8}=2\sqrt{2}$. A simulation of this example is shown in Fig. \ref{fig:Ex3}, where the nominal trajectory is shown in dashed lines.
%width=0.9\columnwidth
\begin{figure}[!hbt]
 \centering
 \includegraphics[scale=0.38,clip]{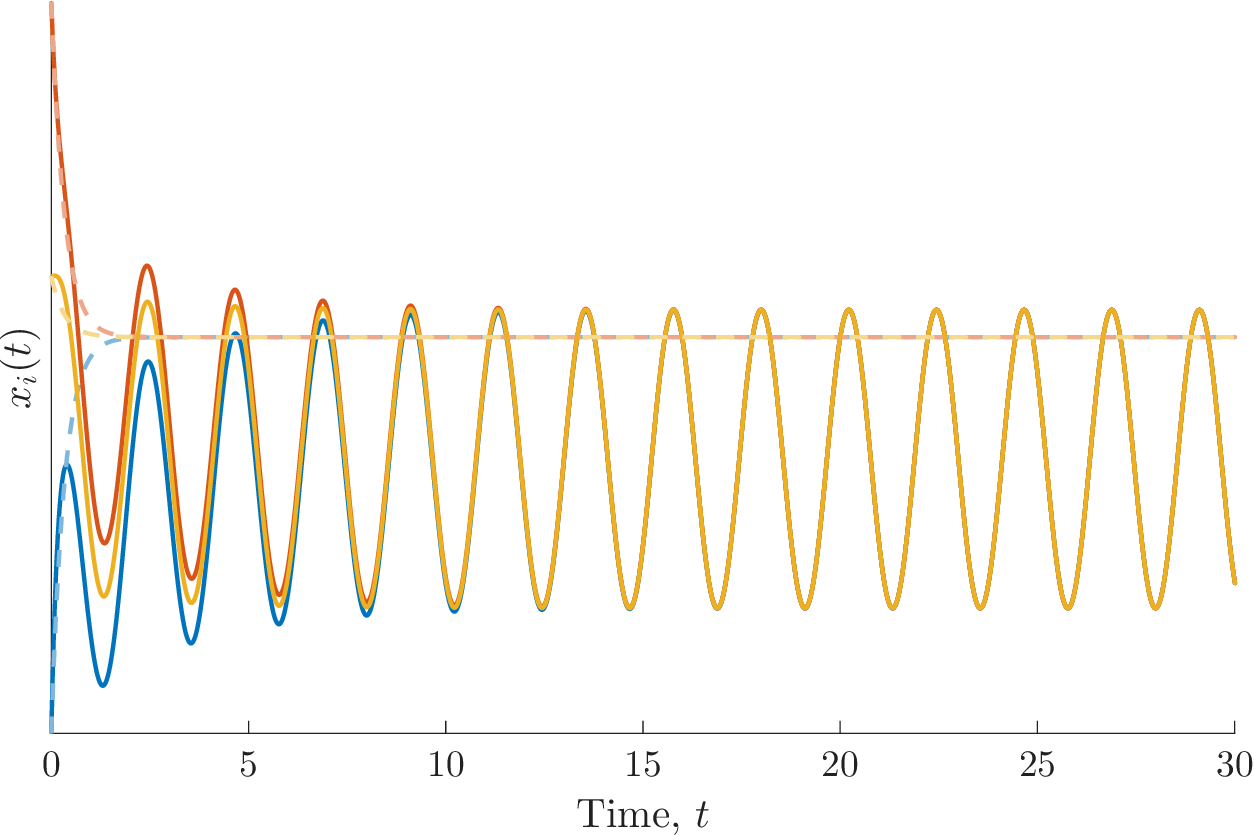}
 \caption{Resulting trajectories for example \ref{sec:ex:int} \ref{sec:ex:int} under \eqref{eq:ex:trans:1}} \label{fig:Ex3}
 \vspace{-.2cm}
\end{figure}
Note that Corollary \ref{cor:UniTrans} is not limited to marginally stable poles and in fact can cause the agreement trajectory to diverge. Consider a first-order stable perturbation
% \[
% \Delta(s)=\frac{k}{s+1},
% \]
%
\begin{equation}\label{eq:ex:trans:2}
\Delta(s)=\frac{k}{s+1} \implies \Lap(s)=\Deg-\frac{k}{s+1}\Adj
\end{equation}
solving for condition $2$ in Corollary \ref{cor:UniTrans} in terms of $k$ and $\lambda$ yields
\[
\lambda+2(1-\Delta(\lambda))=0 \implies \lambda(k)=\frac{-3+\sqrt{1+8k}}{2}
\]
which is a monotonically increasing real function in $k$. In particular $\lambda>0$ for all $k>1$, resulting in a real right half-plane pole in the agreement direction and the only possible agreement trajectories are exponentially diverging ones. An example for $k=3$ (resulting in $\lambda=1$) is shown in Fig.~\ref{fig:Ex3div}.
\begin{figure}[!hbt]
 \centering
    \includegraphics[width=0.9\columnwidth,clip]{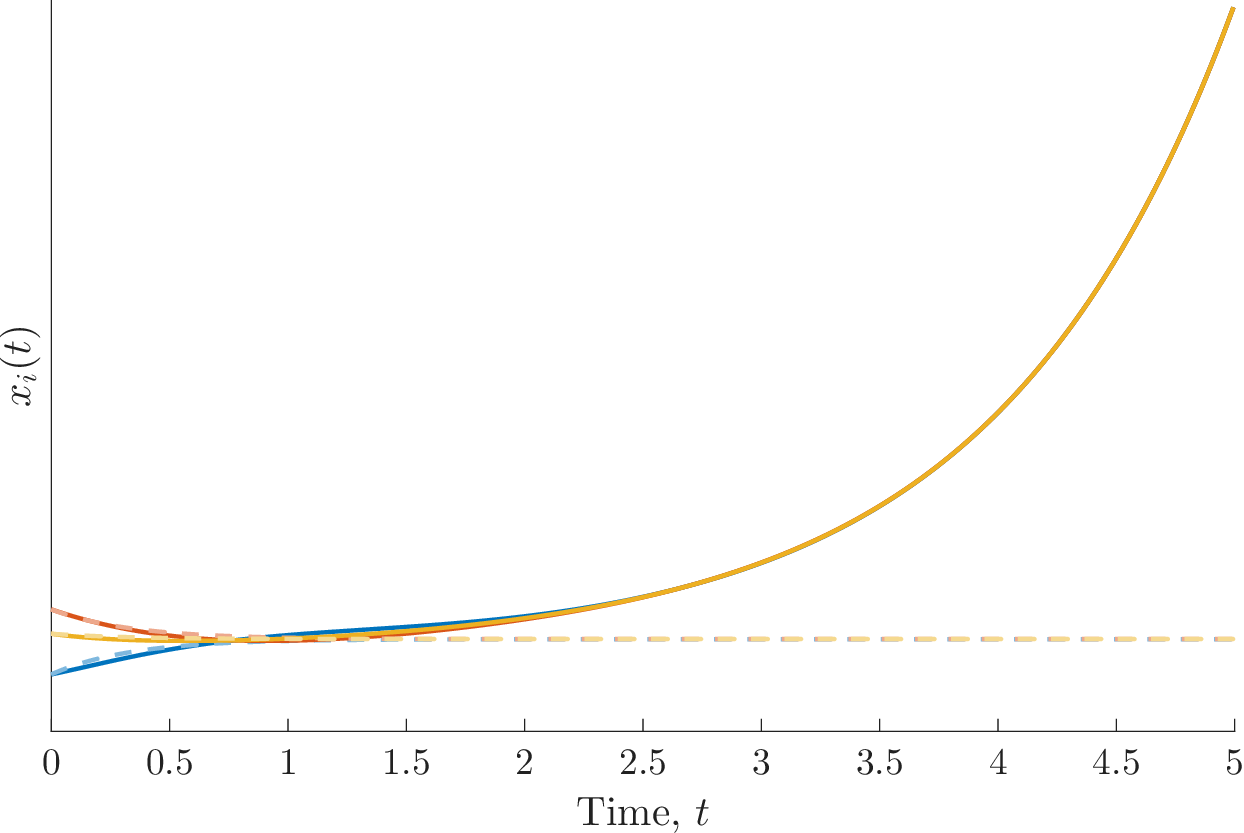}
 \caption{Resulting trajectories for example \ref{sec:ex:int} under \eqref{eq:ex:trans:2}} \label{fig:Ex3div}
 \vspace{-.2cm}
\end{figure}
%
%\todo[inline]{Completely directed graph}
Finally, we illustrate that changes in the agreement trajectory can occur with general graphs as well under the conditions of Theorem \ref{thm:generaltrans}. Consider the same group of integrator agents, not interacting over the directed graph described in Fig.~\ref{fig:Ex:g2} with the following general dynamic Laplacian
\begin{equation}\label{eq:ex:trans:3}
\Lap(s)=\mmatrix{2 &0 &0\\ 0 & 1 &0 \\ 0 &0 & 1}-\mmatrix{0 & \frac{7s+1}{s+3} &\frac{-3s+5}{s+4}\\ \frac{2s}{s+1} & 0 & 0\\ 0 & \frac{3s+1}{s+2} & 0}.
\end{equation}
It is routine to verify that $(\Adj^{(2)}-\Adj(\J))\ones=-\J\ones$ and that $\lambda+\alpha=0\in \spec A$, thus the agents can synchronize to a harmonic trajectory with frequency $1$. In contrast $(\Adj^{(2)}-\Adj(0))\ones\notin \im \ones$, thus they cannot attain consensus. This case is simulated in Fig.~\ref{fig:Ex3gen}.
%width=0.9\columnwidth
\begin{figure}[!hbt]
 \centering
    \includegraphics[scale=0.38,clip]{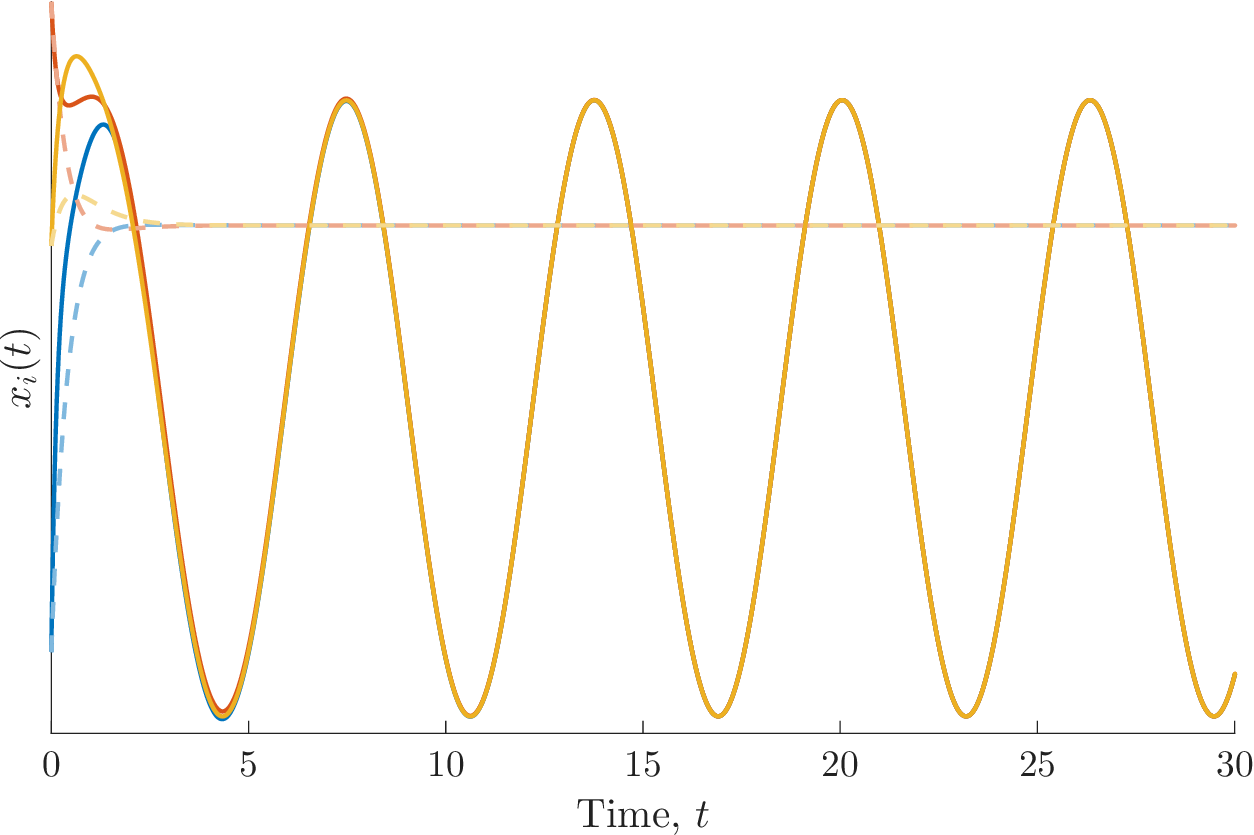}
 \caption{Resulting trajectories for example \ref{sec:ex:int} under \eqref{eq:ex:trans:3}} \label{fig:Ex3gen}
 \vspace{-.2cm}
\end{figure}
\end{subequations}
% A different example of this result would be to a
% Let $F(s)=\E^{-s}$, i.e. a uniform transmission delay with $\tau=1$. Since the example graph is regular with $d=2$, we know that is
% \[
% \J\omega_0+2(1-
% \]
% %
% \begin{figure}[!hbt]
%  \centering
%     \includegraphics[width=0.9\columnwidth,clip]{figs/Ex_unst.eps}
%  \caption{Resulting trajectories for example \ref{sec:ex:unstab}} \label{fig:Ex2}
%  \vspace{-.2cm}
% \end{figure}
% %
% %
% \GB{Consider a system with uniform transmission delay of $1$. The idea here is to pick a $d$-regular graph and design $M$ such that
% %
% \[
%  d(1-\cos(\omega_0\tau))+\J (\omega_0+d\sin(\omega_i \tau))\in \spec M
% \]
% %
% for some values of $\tau$ and $\omega_0$. The point is that $M$ will be unstable (and won't work without the delay), but will synchronize with an appropriate delay.
% }
%
%%%%%%%%%%%%%%%%%%%%%%%%%%%%
\section{Concluding remarks}\label{sec:conc}
%%%%%%%%%%%
In this paper, we introduced a rigorous frequency-domain framework to evaluate the robustness and fragility of synchronization trajectories in LTI multi-agent systems subjected to additional network dynamics. By leveraging coprime factorizations over $\Hinf$, we established a novel Laplace-domain criterion to characterize the specific closed-loop poles that govern the agreement manifold. When restricted to a diffusively coupled system with homogeneous agents, our results simplify further, yielding a clear characterization of the possible agreement trajectories via the image of the dynamic Laplacian. We further showed that the nominal trajectory can be preserved only if the null space of the dynamic Laplacian is preserved when evaluated at the unstable closed-loop poles. Otherwise, additional network dynamics, such as perturbations, can fundamentally alter the asymptotic trajectory.

Motivated by this realization, we analyzed two classes of common dynamic structures: structure-preserving and transmission-only dynamics. As the name suggests, the first class does not alter the null space of the Laplacian, rendering all possible trajectories robust. In contrast, we demonstrated that synchronization trajectories are extremely fragile to transmission-only perturbations under standard LTI cooperative output regulation schemes. In particular, we proved that while static consensus exhibits unique robustness to heterogeneous transmission delays, synchronization to periodic trajectories is destroyed by arbitrarily small transmission dead-time. Furthermore, we showed that for d-regular topologies, uniform transmission perturbations can actively shift the system to synchronize at an entirely unexpected frequency.

Ultimately, these results expose previously unidentified vulnerabilities in classical robust synchronization, demonstrating that transmission dynamics necessitate fundamental structural modifications to networked reference generators. Potential modifications inspired by these findings were briefly discussed in the context of sampled-data systems and are the focus of our ongoing work.
%%%%%%%%%%%%%%%%%%%%%%%%%%%%
%\appendix
\appendices
% \gdef\thesection{\Alph{section}}
% \makeatletter
% \renewcommand\@seccntformat[1]{\appendixname\csname the#1\endcsname.\hspace{0.5em}}
% \makeatother
%%%%%%%%%%%%

\section{Coprime factorizations over $\boldsymbol{\Hinf}$} \label{sec:coprf}

In this Appendix, basic coprime factorization results that are required in the paper are presented. A comprehensive exposition of the subject can be found in \cite{V:85} and extensions to infinite-dimensional systems can be found in \cite{CZw:95}.

Functions $M\in\Hinf^{m\times m}$ and $N\in\Hinf^{p\times m}$ are said to be \emph{right coprime} if there are $X\in\Hinf^{m\times m}$ and $Y\in\Hinf^{m\times p}$ (Bézout coefficients) such that
\begin{subequations} \label{eq:BezE}
\begin{equation} \label{eq:BezE:r}
 XM+YN=I_m.
\end{equation}
Functions $\Mtil\in\Hinf^{p\times p}$ and $\Ntil\in\Hinf^{p\times m}$ are said to be \emph{left coprime} if there are $\Xtil\in\Hinf^{p\times p}$ and $\Ytil\in\Hinf^{m\times p}$
such that
\begin{equation} \label{eq:BezE:l}
\Mtil \Xtil+\Ntil \Ytil=I_p.
\end{equation}
\end{subequations}
A transfer function $G(s)$ is said to have coprime factorizations over \Hinf\ if there are right coprime $M_G,N_G\in\Hinf$ and left coprime $\Mtil_G,\Ntil_G\in\Hinf$, known as right and left coprime factors of $G$, respectively, such that
\begin{equation} \label{eq:Gcf}
 G=N_GM_G^{-1}=\Mtil_G^{-1}\Ntil_G.
\end{equation}
Coprime factors are unique up to post- or pre-multiplication by bi-stable transfer functions for right and left factors, respectively.

\begin{lemma} \label{lem:GstM}%
 If\/ $G(s)$ has coprime factorizations, then
 \[
  G\in\Hinf\iff M_G^{-1}\in\Hinf\iff\tilde M_G^{-1}\in\Hinf.
 \]
\end{lemma}
\begin{IEEEproof}
 The ``if'' part of the first equivalence relation is immediate from \eqref{eq:Gcf}.
 Its ``only if'' part follows from rewriting the Bézout equality \eqref{eq:BezE:r} as
 $M_G^{-1}=X_G+Y_GG$. The second relation follows by similar arguments. 
\end{IEEEproof}

\begin{lemma}\label{lem:coprime}
     Let $G(s)$ have coprime factorizations, then all possible factors satisfy the following properties.
     \begin{enumerate}
         \item If $\lambda \in \CRHP$ is a pole of $G$ then $\Mtil (\lambda)$ and $M(\lambda)$ are singular.
         \item If $G$ is square and with full normal rank and $\lambda \in \CRHP$ is a pole of $G$ but not a zero then $\Ntil(\lambda)$ is regular and $\ker \Mtil(\lambda)=\ker G^{-1}(\lambda)$.
     \end{enumerate}
\end{lemma}
\begin{IEEEproof}
For the first property, assume by contradiction that $\Mtil(\lambda)$ is not singular. Since $\Ntil \in \Hinf$ this implies that $\Mtil^{-1}(\lambda)\Ntil(\lambda)=G(\lambda)$ is bounded and thus $\lambda$ cannot be a pole of $G$. This does not depend on the particular factorization, and the proof for $M$ is identical. For the second, note that since $G$ is square and with full normal rank by definition so is $\Ntil$, and $G^{-1}$ exists for almost all $s$.  By \cite[Lemma 7.2.21]{CZw:95}, for $\lambda\in \CRHP$, $\Ntil(\lambda)$ loses rank if and only if $\lambda$ is a zero of $G$. Since $\Ntil$ is square and $\lambda$ is by assumption not a zero of $G$, $\ker \Ntil(\lambda)$ is trivial. The las points follows immediately by rewriting $\Mtil=\Ntil G^{-1}$.
\end{IEEEproof}
%%%%%%%%%%%%%%%%%%%%%%%%%%%%
%%%%%%%%%%%%%%%%%%%%%%%%%%%%
%\section*{References}
%%%%%%%%%%%%%%%%%%%%%%%%%%%%
%\vspace{-0.50cm}

\bibliographystyle{IEEEtran}
%\bibliography{papers,books}
\bibliography{MyBib}

@preamble{"\newcommand{\noopsort}[1]{}"}

@string{ACC = "American Control Conf."}

@string{Aut = "Automatica"}

@string{ECC = "European Control Conf."}

@string{IEEE-TAC = "IEEE Trans.\ Automat.\ Control"}

@string{IEEE-TCNS = "IEEE Trans.\ Control Netw.\ Syst."}

@string{IEEE-CDC = "IEEE Conf.\ Decision and Control"}

@string{MCSS = "Math.\ Control, Signals and Systems"}

@string{PrIEEE = "Proc.\ IEEE"}

@string{SIAM-CO = "SIAM J.\ Control Optim."}

@string{IEEE-TSMC = "{IEEE} Trans.\ Syst.\, Man, and Cybernetics: Systems"}

@string{SV    = "Springer-Verlag"}

@string{MIT  = "The MIT Press"}

@string{PUP  = "Princeton University Press"}

@string{SV   = "Springer-Verlag"}

@article{BenRejeb2018GuaranteedCost,
  author  = {Ben Rejeb, Jihene and Mor{\u a}rescu, Irinel-Constantin and Daafouz, Jamal},
  title   = {Control Design with Guaranteed Cost for Synchronization in Networks of Linear Singularly Perturbed Systems},
  journal = {Automatica},
  volume  = {91},
  pages   = {89--97},
  year    = {2018},
  doi     = {10.1016/j.automatica.2018.01.019}
}

@book{ME:10,
  author = "M. Mesbahi and M. Egerstedt",
  title = "Graph Theoretic Methods in Multiagent Networks",
  publisher = PUP,
  address = "Princeton",
  year = "2010"
}

@article{FM:04,
  author = "J. A. Fax and R. M. Murray",
  title = "Information Flow and Cooperative Control of Vehicle Formations",
  journal = IEEE-TAC,
  volume = "49",
  year = "2004",
  number = "9",
  pages = "1465--1476"
}

@ARTICLE{M:05,
  author={L. {Moreau}},
  journal=IEEE-TAC, 
  title={Stability of multiagent systems with time-dependent communication links}, 
  year={2005},
  volume={50},
  number={2},
  pages={169-182}
}

@article{O-SFM:07,
  author = "R. Olfati-Saber and A. Fax and R. M. Murray",
  title = "Consensus and Cooperation in Networked Multi-Agent Systems",
  journal = PrIEEE,
  volume = "95",
  year = "2007",
  number = "1",
  pages = "215--233"
}

@book{CZw:95,
  author = "R. F. Curtain and H. Zwart",
  title = "An Introduction to Infinite-Dimensional Linear Systems Theory",
  publisher = SV,
  address = "New York, NY",
  year = "1995"
}

@ARTICLE{LDCH:10,
  author={Li, Zhongkui and Duan, Zhisheng and Chen, Guanrong and Huang, Lin},
  journal={IEEE Transactions on Circuits and Systems I: Regular Papers}, 
  title={Consensus of Multiagent Systems and Synchronization of Complex Networks: A Unified Viewpoint}, 
  year={2010},
  volume={57},
  number={1},
  pages={213-224},
}

@article{Sm:89,
  author = "M. C. Smith",
  year = "1989",
  title = "On stabilization and the existence of coprime factorizations",
  journal = IEEE-TAC,
  volume = "34",
  number = "9",
  pages = "1005--1007"
}

@article{WSA:11,
  author = "P. Wieland and R. Sepulchre and F. Allg{\"o}wer",
  title = "An internal model principle is necessary and sufficient for linear output synchronization",
  journal = Aut,
  volume = "47",
  year = "2011",
  number = "5",
  pages = "1068--1074"
}

@article{GSm:93,
  author = "T. T. Georgiou and M. C. Smith",
  title = "Graphs, causality and stabilizability: linear, shift-invariant systems on {$\mathcal{L}_2[0,\infty)$}",
  journal = MCSS,
  volume = "6",
  year = "1993",
  pages = "195--223"
}

@book{V:85,
  author = "M. Vidyasagar",
  title = "Control System Synthesis: A Factorization Approach",
  publisher = "The MIT Press",
  address = "Cambridge, MA",
  year = "1985"
}

@article{SS:09,
  author = "L. Scardovi and R. Sepulchre",
  title = "Synchronization in networks of identical linear systems",
  journal = Aut,
  volume = "45",
  year = "2009",
  number = "11",
  pages = "2557--2562"
}

@article{IMC:14,
  author = "A. Isidori and L. Marconi and G. Casadei",
  title = "Robust Output Synchronization of a Network of Heterogeneous Nonlinear Agents Via Nonlinear Regulation Theory",
  journal = IEEE-TAC,
  volume = "59",
  year = "2014",
  number = "10",
  pages = "2680--2691"
}

@article{BMZ:23,
title = {On the internal stability of diffusively coupled multi-agent systems and the dangers of cancel culture},
journal = {Automatica},
volume = {155},
pages = {111158},
year = {2023},
author = {G. Barkai and L. Mirkin and D. Zelazo},
}

@INPROCEEDINGS{BMZ:23cdc,
  author="G. Barkai and L. Mirkin and D. Zelazo",
  booktitle="Proc.\ 62nd "#IEEE-CDC,
  title={An emulation approach to sampled-data synchronization}, 
  year={2023},
  address = "Singapore",
  pages={6449-6454},
}

@INPROCEEDINGS{BMZ:24cdc,
title = "Asynchronous sampled-data synchronization with small communications delays",
author = {G. Barkai and L. Mirkin and D. Zelazo},
booktitle="Proc.\ 63rd IEEE Conf.\ Decision and Control.", 
 year={2024},
}

@INPROCEEDINGS{BMZ:24ECC,
  author={G. Barkai and L. Mirkin and D. Zelazo},
  booktitle="Proc.\ 22nd European Control Conf.", 
  title={An Emulation Approach to Output-Feedback Sampled-Data Synchronization}, 
  year={2024},
 }

@article{FrW:76,
  author = "B. A. Francis and W. M. Wonham",
  title = "The internal model principle of control theory",
  journal = Aut,
  volume = "12",
  year = "1976",
  number = "5",
  pages = "457--465"
}

@INPROCEEDINGS{BMZ:25,
  author="G. Barkai and L. Mirkin and D. Zelazo",
  booktitle="Proc.\ 24th "#ECC,
  title={On two-degrees-of-freedom agreement protocols}, 
  year={2026},
  address = "Reykjavík, Iceland",
  pages={1492-1497},
}

@article{At:14,
author = {F.M. Atay},
title = {The consensus problem in networks with transmission delays},
journal = {Philosophical Transactions of the Royal Society A: Mathematical, Physical and Engineering Sciences},
volume = {371},
number = {1999},
pages = {20120460},
year = {2013},
}

@article{O-SM:04,
  author = "R. Olfati-Saber and R. M. Murray",
  title = "Consensus Problems in Networks of Agents With Switching Topology and Time-Delays",
  journal = IEEE-TAC,
  volume = "49",
  year = "2004",
  number = "9",
  pages = "1520--1533"
}

@INPROCEEDINGS{SDJ:08,
  author={A. Seuret and D.V. Dimarogonas and K.H. Johansson},
  booktitle="Proc.\ 47th "#IEEE-CDC,
  title={Consensus under communication delays}, 
  year={2008},
  pages={4922-4927},
  }

@article{MPA:10,
title = {Delay robustness in consensus problems},
journal = Aut,
volume = {46},
number = {8},
pages = {1252-1265},
year = {2010},
author = {U. Münz and A. Papachristodoulou and F. Allgöwer},
}

@article{MMN:09,
author = {W. Michiels and C-I. Mor\u{a}rescu and S-I. Niculescu},
title = {Consensus Problems with Distributed Delays, with Application to Traffic Flow Models},
journal = SIAM-CO,
volume = {48},
number = {1},
pages = {77-101},
year = {2009},
}

@INPROCEEDINGS{GM:04acc,
  author={A. Gattami and R. Murray},
  booktitle="Proc.\ 2004 "#ACC,
  title={A frequency domain condition for stability of interconnected MIMO systems}, 
  year={2004},
  volume={4},
  pages={3723-3728 vol.4},
 }

@ARTICLE{HT:13,
  author={J.M. Hendrickx and J.N. Tsitsiklis},
  journal=IEEE-TAC, 
  title={Convergence of Type-Symmetric and Cut-Balanced Consensus Seeking Systems}, 
  year={2013},
  volume={58},
  number={1},
  pages={214-218},
  }

@ARTICLE{LC:17,
  author={Z. Li and J. Chen},
  journal=IEEE-TAC,
  title={Robust Consensus of Linear Feedback Protocols Over Uncertain Network Graphs}, 
  year={2017},
  volume={62},
  number={8},
  pages={4251-4258},
  }

@ARTICLE{TTM:13,
  author={H.L. Trentelman and K. Takaba and N. Monshizadeh},
  journal=IEEE-TAC, 
  title={Robust Synchronization of Uncertain Linear Multi-Agent Systems}, 
  year={2013},
  volume={58},
  number={6},
  pages={1511-1523},
  }

@ARTICLE{QGY:19,
  author={W. Qian and Y. Gao and Y. Yang},
  journal=IEEE-TSMC, 
  title={Global Consensus of Multiagent Systems With Internal Delays and Communication Delays}, 
  year={2019},
  volume={49},
  number={10},
  pages={1961-1970},
 }

@ARTICLE{Fr:77,
  author={B. Francis},
  journal=IEEE-TAC,
  title={The multivariable servomechanism problem from the input-output viewpoint}, 
  year={1977},
  volume={22},
  number={3},
  pages={322-328},
 }

@ARTICLE{ZB:17,
  author={D. Zelazo and M. Bürger},
  journal=IEEE-TCNS,
  title={On the Robustness of Uncertain Consensus Networks}, 
  year={2017},
  volume={4},
  number={2},
  pages={170-178},
}

@article{DG:75,
title = {Robust control of a general servomechanism problem: The servo compensator},
journal = Aut,
volume = {11},
number = {5},
pages = {461-471},
year = {1975},
author = {E.J. Davison and A. Goldenberg},
}

@article{Lee:16,
author = {Lee, Dongjun},
title = {Robust consensus of linear systems on directed graph with non-uniform delay},
journal = {IET Control Theory \& Applications},
volume = {10},
number = {18},
pages = {2574-2579},
year = {2016}
}

@article{LL:17unk,
author = {M. Lu and L. Liu},
title = {Consensus of linear multi-agent systems subject to communication delays and switching networks},
journal = {International Journal of Robust and Nonlinear Control},
volume = {27},
number = {9},
pages = {1379-1396},
year = {2017}
}

@ARTICLE{LL:17know,
  author={M. Lu and L. Liu},
  journal=IEEE-TAC, 
  title={Distributed Feedforward Approach to Cooperative Output Regulation Subject to Communication Delays and Switching Networks}, 
  year={2017},
  volume={62},
  number={4},
  pages={1999-2005},
 }

\end{document}